\documentclass[a4paper,twocolumn,11pt]{quantumarticle}
\pdfoutput=1
\usepackage[utf8]{inputenc}
\usepackage[english]{babel}
\usepackage[T1]{fontenc}
\usepackage{amsmath}
\usepackage{hyperref}

\usepackage{tikz}
\usepackage{lipsum}
\usepackage{amsfonts}
\usepackage{braket}
\usepackage{amssymb}
\usepackage{amsthm}
\usepackage[numbers,sort&compress]{natbib}
\usepackage[ruled]{algorithm2e}
\usepackage{multirow}
\newtheorem{theorem}{Theorem}
\newtheorem{definition}{Definition}
\newtheorem{lemma}{Lemma}
\newtheorem{corollary}{Corollary}

\begin{document}

\title{Quantum polar stabilizer codes based on polarization of pure quantum channel don't work for quantum computing}

\author{Zhengzhong Yi}
\affiliation{Harbin Institute of Technology, Shenzhen, China}
\author{Zhipeng Liang}
\affiliation{Harbin Institute of Technology, Shenzhen, China}
\author{Yulin Wu}
\affiliation{Harbin Institute of Technology, Shenzhen, China}
\author{Xuan Wang}
\affiliation{Harbin Institute of Technology, Shenzhen, China}
\email{wangxuan@cs.hitsz.edu.cn}

\maketitle

\begin{abstract}
 Inspired by classical polar codes, whose coding rate can asymptotically achieve the Shannon capacity, researchers are trying to find its analogue in quantum information field, which are called quantum polar codes. However, no one has designed a quantum polar coding scheme which applies to quantum computing yet. There are two intuitions in previous research. The first is that directly converting classical polar coding circuits to quantum ones will produce polarization phenomenon of pure quantum channel, which has been proved in our previous work. The second is that based on this quantum polarization phenomenon one can design a quantum polar coding scheme that applies to quantum computing. There are several previous work following the second intuition, none of which has been verified by experiments. In this paper, we follow the second intuition and propose a more reasonable quantum polar stabilizer code construction algorithm than any previous ones by using the theory of stabilizer codes. Unfortunately, simulation experiments show that even the stabilizer codes obtained from this more reasonable construction algorithm don't work, which implies that the second intuition leads to a dead end.  Based on the analysis on why the second intuition don't work, we provide a possible future direction of designing quantum stabilizer codes with high coding rate by borrowing the idea of classical polar codes. following this direction, we find a class of quantum stabilizer codes with coding rate 0.5 for pure Pauli X, Z and Y noise.
\end{abstract}

\section{Introduction}
Due to the great potential to solve certain problems which cannot be solved by classical computer with reasonable resources and time, quantum computing has attracted worldwide attention in the past twenty years.  Nowadays, researchers have studied many physical systems which can be used to realize quantum computing, including superconducting circuits\cite{arute2019quantum,wu2021strong,gong2021quantum}, trapped ions\cite{bruzewicz2019trapped}, and quantum dots\cite{hendrickx2021four}. However, limited by the hardware, none of these systems has realized reliable large-scale quantum computing, since qubits in these physical systems are easily to be influenced by the environment, which makes the state of qubits and hence the quantum information carried by them prone to error. Fortunately, quantum error correcting (QEC) proposed by Shor and Steane\cite{PhysRevA.52.R2493,PhysRevLett.77.793} provide us with a solution to this problem.

Though quantum error correcting codes (QECCs) provide us a promising approach for fault-tolerant large-scale quantum computing, the high qubit overhead results from low coding rate of most QECCs challenges current hardware level. Similar to classical error correcting codes (CECCs), QECCs encode $n$ (which is called code length) less reliable physical qubits (with error rate $p$)  to obtain $k$ more reliable logical qubits (with error rate $p_{L}<p$ after decoding and recovery). The ratio $k/n$ is called coding rate. Improving coding rate is beneficial to reducing qubit overhead. 

No matter for CECCs or QECCs, in general, we often increase the code length $n$  to improve the reliability of the logical bits/qubits. Good CECCs have constant or even asymptotically increasing coding rate with the increase of code length $n$\cite{1057683,748992,mackay1996near,5075875}. For instance, experimental results show that the coding rate of classical low-density parity check (CLDPC) codes\cite{1057683,748992,mackay1996near} seems to achieve Shannon capacity. Another example is that Arikan has proved the coding rate of classical polar code\cite{5075875} can asymptotically achieve the Shannon capacity.

These CECCs with high coding rate have inspired researchers to look for efficient QECCs. For example, inspired by CLDPCs, researchers proposed quantum low-density parity check (QLDPC) codes  whose coding rate can maintain constant\cite{gottesman2013fault,tillich2013quantum,freedman2013quantum,guth2014quantum,hastings2013decoding,breuckmann2016constructions,breuckmann2017hyperbolic,breuckmann2021single,grospellier2021combining}. However, the relationship between the coding rate of QLPDCs and quantum channel capacity is still not clear. If we take maximum single letter coherent information (MSLCI), which will be explained in Section \ref{2.2}, as quantum channel capacity, the asymptotic coding rate of certain QLDPCs seems to rather lower than quantum channel capacity\cite{yi2023channel}. Another example of efficient QECCs research inspired by CECCs is quantum polar codes inspired by classical polar codes\cite{5075875}. In previous work\cite{guo2013polar,6284203,6302198,renes2012efficient,wilde2013polar,renes2015efficient,hirche2015polar,7370934,8815775,goswami2021quantum}, some quantum polar coding schemes have been proposed. However, either they are based on classical-quantum channel\cite{guo2013polar,6284203,6302198,wilde2013polar,hirche2015polar,7370934,goswami2021quantum}, or they ask for destructive measurements\cite{renes2012efficient,renes2015efficient} which refers to the measurements that destruct the quantum information carried by quantum states. These makes them don’t apply to quantum computing where the computing channel is pure quantum channel and destructive measurement should be avoided.

There are two intuitions in current research on quantum polar coding\cite{guo2013polar,6284203,6302198,renes2012efficient,wilde2013polar,renes2015efficient,hirche2015polar,7370934,8815775,goswami2021quantum,8989387,9366784}. The first is that directly converting the classical polar coding circuits into quantum ones will lead to polarization phenomenon of quantum channels. The second is that based on polarization phenomenon of quantum channels arises from directly converting the classical polar coding circuits into quantum ones, one can design a quantum polar coding scheme for quantum computing. In our previous work\cite{yi2023channel}, we have proved the first intuition in two-dimensional-input quantum symmetric channels (QSCs). There are several previous work following the second intuition\cite{guo2013polar,renes2012efficient,renes2015efficient,8989387,9366784}. However, none of them is a quantum stabilizer code for quantum computing in a real sense except the one in Ref. \cite{9366784} — either they are based on classical-quantum channels\cite{guo2013polar} or they ask for destructive measurement\cite{renes2012efficient,renes2015efficient}. In Ref. \cite{9366784}, the authors give a method to construct the so called "entanglement assisted quantum polar stabilizer codes". However, there is no experiment to support the validity of these codes. In this paper, we follow the second intuition and propose two construction algorithm called coherent-information-achieving (CA) algorithm and block-selection (BS) algorithm. CA is the same as the construction in Ref. \cite{9366784}, while we believe BS is more reasonable than CA because it's designed by using the theory of stabilizer codes. By saying "using the theory of stabilizer codes", we mean that the more reasonable construction has considered the weight of logical operators and stabilizers.  Unfortunately, simulation experiments show that even the stabilizer codes obtained from this more reasonable construction algorithms called BS algorithm don't work, which implies that the second intuition leads to a dead end.

Quantum stabilizer codes are an important class of QECCs. They avoid destructive measurements in error symdrome detection by measuring stabilizers. Different from the decoders of CECCs, which decode the value of the received bits into the value of bits before encoding, the decoders of quantum stabilizer codes decode the error syndromes which are composed of the measurement result of stabilizers into the errors happen to each qubit. Quantum stablizer codes can be represented by their stabilizer generators. Determining a stabilizer generators construction method is equivalent to determing a kind of quantum stabilizer codes. According to code construction of classical polar code and the ranking of quantum coordinate channels, we propose two algorithms for constructing stabilizer generators of QPSCs. The first one ranks the quantum coordinate channels by the Bhattacharyya parameters. Since it allows the coding rate achieving the channel capacity if we take MSLCI as quantum channel capacity, we name it coherent-information--achieving (CA) construction algorithm. The second one ranks the quantum coordinate channels by the weight of their corresponding stabilizer generators. Since this algorithm rank the coordinate channels in a block-by-block way, we name it block-selection (BS) construction algorithm. 

In classical polar codes, the frozen bits are frozen to a fixed value such as zero before encoding, which is a priori knowledge for decoders. In decoding, decoders will use this priori knowledge and decode all frozen bits into the fixed value. Hence, in the whole encoding-transmitting-decoding process, the frozen bits seem to suffer no error. This point is crucial to guarantee the reliability of polar codes, since the decoding of frozen bits will influence the decoding of information bits. For QPSCs, there are two kinds of assumptions about the frozen qubits -- unreliable frozen qubits assumption and reliable frozen qubits assumption. The unreliable frozen qubits assumption assumes that the frozen qubits suffer from error as the infromation qubits does, which is the common assumption in error correction of quantum computing, while the reliable frozen qubits assumption assumes that the frozen qubits don't suffer from error, which might be too strong for quantum computing. 

To test the error correcting capability of QPSCs under these two construction algorithms, we perform simulation in bit-flip channel under the above two assumptions, respectively. The simulation results show that, under both assumptions, for both CA and BS construction algorithm, in the range of $p$ from $1 \times 10^{-5}$ to $1 \times 10^{-2}$, no evidence shows that the logical error rate (LER) of single logical qubit $LER_{lq}$ can be decreased by increasing code length and we cannot find the noise threshold. That is to say, these QPSCs don't work for quantum computing. Therefore, we come to a conclusion that the second intuition is too naive to guide us to design quantum polar coding scheme for quantum computing. 

The rest of paper is organized as follows. Some preliminaries are introduced in Section \ref{2}, including stabilizer codes, quantum channel and its capacity, classical polar code and its code construction, and polarization of two-dimensional-input quantum symmetric channels. In Section \ref{3} we will reveal the relationship between the channel quality of quantum coordinate channels and the weight of stabilizer generators and logical operators of QPSCs. In Section \ref{4}, we give out CA and BS construction algorithms of QPSCs. In Section \ref{5}, simulation results are given and analyzed. In Section \ref{6}, we provide a possible future direction of designing quantum stabilizer does with high coding rate by borrowing the idea of classical polar codes. We conclude our work in Section \ref{7}.

\section{Preliminaries}
\label{2}
\subsection{Stabilizer codes}
\label{2.1}

This section briefly introduces quantum stabilizer codes. Quantum stabilizer codes are an important class of QECCs, whose code construction is analogous to classical linear codes\cite{nielsen2002quantum}.

In quantum information theory, errors acting on a single qubit are generally modeled as elements of the Pauli group denoted by $\mathcal{G}_{1}$, which consists of all Pauli operators together with multiplicative constants $\{\pm1,\pm i\}$, namely,
\begin{equation}
	\mathcal{G}_{1}=\left\{\pm I,\pm iI,\pm X,\pm iX,\pm Y,\pm iY,\pm Z,\pm iZ\right\}
\end{equation}
The general Pauli group $\mathcal{G}_{n}$ on $n$ qubits is the $n$-fold tensor product of $\mathcal{G}_{1}$, namely, $\mathcal{G}_{n}=\mathcal{G}_{1}^{\otimes n}$. Here, we use notation $w\left(\cdot\right)$ to denote the weight of an operator $E\in\mathcal{G}_{n}$, namely, the number of qubits on which it acts nontrivially (with a Pauli operator other than the identity $I$). For instance, $w\left(X_1X_2I_3\right)=2$.

A quantum stabilizer code $C$ encodes $k$ logical qubits into $n$ physical qubits. Denote the code space by $\mathcal{Q}_{C}$ which is a subspace of ${(\mathbb{C}^{2})}^{\otimes n}$. The code space $\mathcal{Q}_{C}$ is stabilized by a stabilizer group $\mathcal{S}$ which is a subgroup  of $\mathcal{G}_{n}$. Giving the stabilizer group $\mathcal{S}$ of a code C is equivalent to explicitly giving its code space $\mathcal{Q}_C$. Using the stabilizer group $\mathcal{S}$, the code space $\mathcal{Q}_C$ can be defined as
\begin{equation}
	\begin{aligned}
		\mathcal{Q}_C=\{\ket{\psi}\in\left(\mathbb{C}^2\right)^{\otimes n}:\ S\ket{\psi}=\ket{\psi},\forall\ S\in\mathcal{S}\}
	\end{aligned}
\end{equation}

A group $G$ can be generated by a set of independent elements $g_1,\cdots,g_l$, which means each element in $G$ can be written as a product of elements from the list $g_{1},\cdots,g_l$. We say $g_1,\cdots,g_l$ are generators of group $G$, and we write $G=\langle g_1,\cdots,g_l\rangle$.

Similar to classical linear codes, a stabilizer code $C$ can be expressed as a 3-tuple $\left[ n, k, d \right]$, which means it encodes $k$ logical qubits into $n$ physical qubits and its code distance is $d$. For an $\left[ n, k, d \right]$ code $C$, its stabilizer group $\mathcal{S}$ has $n-k$ independent generators, namely, $\mathcal{S}=\langle S_1,\cdots,S_{n-k}\rangle$. Besides, we can also find another subgroup $\mathcal{L}$ of $\mathcal{G}_n$ that has $k$ pairs of generators, namely,$\mathcal{L}=\langle \overline{X}_1,\overline{Z}_1,\cdots,\overline{X}_k,\overline{Z}_k\rangle$, whose elements commute with all elements in $\mathcal{S}$\cite{RevModPhys.87.307}. The elements in group $\mathcal{L}$ are called the logical operators of code $C$. Code distance $d$ is defined as the minimum weight of the logical operators, namely,
\begin{equation}
	\begin{aligned}
		d=\min\limits_{P\in\mathcal{L}}{w(P)}
	\end{aligned}
\end{equation}

For a $\left[n, k\right]$ stabilizer code $\mathcal{C}$, its encoding process can be depicted by a unitary operator $U_{enc}$ acting on $n$ qubits. These $n$ qubits can be divided into two different sets. The first set consists of $k$ \emph{data qubits} (or called logical qubits) which contains information to be encoded, and the second set consists of $n-k$ \emph{ancilla qubits} which are set to a fixed state such as $\ket{0}^{\otimes n-k}$. The Pauli operators acting on $k$ logical qubits before encoding will be transformed to $k$ sets of logical operators by the encoding operator $U_{enc}$ after encoding, namely, $\overline{L}_i=U_{enc}P_iU_{enc}^\dagger$, $1\le i\le k$. The Pauli operators acting on $n-k$ ancilla qubits are transformed to stabilizer generators by the encoding operator $U_{enc}$ after encoding, namely, $S_i=U_{enc}P_i U_{enc}^\dagger$, $k+1\le i\le n$.

\subsection{Pure quantum channel and its capacity}
\label{2.2}
\subsubsection{Pure quantum channel}
\label{2.2.1}

If the inputs and outputs of a channel are both arbitrary quantum states, it is a pure quantum channel. The behavior of a quantum channel $\mathcal{E}$ acting on a quantum system $P$ can be described by a set of operator elements $\{ E_i \}_{i=1,\cdots n}$, which satisfies $\sum_{i=1}^{n}{E_i^\dag E_i}=I$\cite{nielsen2002quantum}. Suppose the state of quantum system $P$ is $\rho^P$, and $P$ is subject to the quantum channel $\mathcal{E}$, which changes system $P$ to $P^\prime$ and maps the state $\rho^P$ to $\rho^{P^\prime}$. Using operator-sum representation\cite{nielsen2002quantum}, the transformation of state $\rho^P$ can be described as
\begin{equation}
	\begin{aligned}
		\rho^{P^\prime}=\mathcal{E}\left(\rho^P\right)=\sum_{i=1}^{n}{E_i\rho^PE_i^\dag}
	\end{aligned}
\end{equation}

Common quantum channels include bit-flip channel, phase-flip channel and depolarizing channel, 
whose operator elements are $\{\sqrt p X,\sqrt{1-p}I\}$, $\{\sqrt p Z,\sqrt{1-p}I\}$, and $\{\sqrt{p/3} X,\sqrt{p/3} Y, \sqrt{p/3} Z, \sqrt{1-p}I\}$, respectively. 

In our previous work\cite{yi2023channel}, we propose a mathematical tool, \textbf{basis transition probability matrix (BTPM)}, which is inspired by transition probability matrix (TPM) of classical channel, to depict the behavior of quantum channels. Besides, we prove that, for a quantum channel $\mathcal{E}$ having BTPM, giving its BTPM is equivalent to giving its operator elements $\{ E_i \}_{i=1,\cdots n}$, which satisfies $\sum_{i=1}^{n}{E_i^\dag E_i}=I$. 

In \cite{yi2023channel} we also define two classes of quantum channel, \textbf{quantum symmetric channel (QSC)} and \textbf{quantum quasi symmetric channel (QQSC)}. For a quantum channel which has BTPM, if each row of the BTPM is a permutation of the first row, then this quantum channel is symmetric with respect to its input. If each column of the BTPM is a permutation of the first column, then this quantum channel is symmetric with respect to its output. If a quantum channel is symmetric with respect to both of its input and output, then this channel is called QSC. If a channel is symmetric with respect to its input but might not to its output, and its BTPM can be divided into several submatrices by column, each of which satisfies that each column of it is a permutation of the first column of it, then this channel is called a QQSC. 

\subsubsection{Quantum channel capacity}
\label{2.2.2}
Coherent information of a quantum channel $\mathcal{E}$, denoted by $I(\rho^P,\mathcal{E})$, is the quantity measuring the amount of quantum information transmitted in the channel\cite{PhysRevA.54.2629},  where  $\rho^P$ is the state of quantum system $P$ before being transformed  by $\mathcal{E}$. Coherent information is believed to be the analogue of classical mutual information in quantum information theory\cite{nielsen2002quantum} and a quantity to measure the channel capacity of pure quantum channels\cite{9241807,8242350,holevo2020quantum,5592851,holevo2010mutual,bennett2004quantum,PhysRevA.57.4153,PhysRevA.55.1613,kretschmann2004tema,shor2003capacities}. However, since the coherent information has superadditivity, it's still not clear how to calculate quantum channel capacity presicely by coherent information. Here, we use maximum single letter coherent information (MSLCI)\cite{yi2023channel} to measure quantum channel capacity, which is believed to be a lower bound of the average coherent information of multiple use of a quantum channel\cite{leditzky2018dephrasure}.

For a classical symmetric channel, Arikan has proved that its symmetric capacity is it Shannon capacity. Similarly, in \cite{yi2023channel} we prove that the MSLCI of a two-dimensional-input QQSC equals to its symmetric coherent information whose definition is as follows.

\begin{definition}[\textbf{Symmetric coherent information}]
	\label{Symmetric coherent information}
	For a quantum channel $\mathcal{E}$, the number of whose input qubits is n, its input state can be represented by $\rho=\sum_{i=1}^{2^{n}} q_{i}\ket{i}\bra{i}$, its symmetric information $I_U$ is defined as the coherent information $I(\rho,\mathcal{E})$ when $q_1=\cdots=q_{2^n}=\frac{1}{2^n}$, namely,
	\begin{equation}
		I_{U} \equiv I\left(\rho=\sum_{i=1}^{2^{n}} \frac{1}{2^{n}}\ket{i}\bra{i}, \mathcal{E}\right)
	\end{equation}
\end{definition}

\subsection{Classical polar code and its code construction}
\label{2.3}
For a binary-input discrete memoryless channel (B-DMC) $W: \mathcal{X}\rightarrow\mathcal{Y}$ with input alphabet $\mathcal{X}$, output alphabet $\mathcal{Y}$ and transition probabilities $W\left(y\middle| x\right)$, $x\in\mathcal{X}=\{0,1\}$, $y\in\mathcal{Y}$, Arikan defines two parameters\cite{5075875}, the symmetric capacity and the Bhattacharyya parameter, to measure its quality. The symmetric capacity of a B-DMC is defined as,
\begin{equation}
	\begin{aligned}
		&I(W)\\
		& \triangleq \sum_{y \in \mathcal{Y}} \sum_{x \in \mathcal{X}} \frac{1}{2} W(y \mid x) \log \frac{W(y \mid x)}{\frac{1}{2} W(y \mid 0)+\frac{1}{2} W(y \mid 1)}
	\end{aligned}
\end{equation}
and its Bhattacharyya parameter is defined as,
\begin{equation}
	Z(W) \triangleq \sum_{y \in \mathcal{Y}} \sqrt{W(y \mid 0) W(y \mid 1)}
\end{equation}

Arikan has proved that the symmetric capacity $I(W)$ is equal to the Shannon capacity when $W$ is a symmetric channel. One can see $Z(W)$ takes values in $\left[0,1\right]$, and $I(W)$ also take values in $\left[0,1\right]$ if we use base-2 logarithms. When $I(W)$ tends to $1$, $Z(W)$ tends to $0$, and when $I(W)$ tends to $0$, $Z(W)$ tends to $1$.

Classical polar coding proposed by Arikan\cite{5075875} is the only coding scheme whose coding rate is proved to achieve the symmetric capacity of any given B-DMC. This scheme is based on the channel polarization which consists of two processes, channel combining and channel splitting.

As shown in Fig. \ref{fig1}, in channel combining, we combine $N$ copies of a given B-DMC $W: \mathcal{X}\rightarrow\mathcal{Y}$ (whose symmetric capacity is $I(W)$) with input alphabet $\mathcal{X}$ and output alphabet $\mathcal{Y}$ in a recursive manner and obtain a combined channel $W_N: \mathcal{X}^N\rightarrow\mathcal{Y}^N $ ($\mathcal{X}^N$ and $\mathcal{Y}^N$ are the $N$-power extension alphabet of $\mathcal{X}$ and $\mathcal{Y}$, respectively) which maps the vector $(x_1,\cdots,x_N)$ to $(y_1,\cdots,y_N)$. Here we will use $x_1^N$ as a shorthand for vector $(x_1,\cdots,x_N)$. In channel splitting, the combined channel $W_N$ is split back into a set of $N$ binary-input coordinate channels $\{{W_N^{(i)}}\}_{1\le i\le N}$. The coordinate channel $W_N^{(i)}$ is define as,
\begin{equation}
	W_{N}^{(i)}: \mathcal{X} \rightarrow \mathcal{Y}^{N} \times \mathcal{X}^{i-1}, 1 \leq i \leq N
\end{equation}

As $N$ becomes large, the symmetric capacity of $N\times I(W)$ coordinate channels will asymptotically tend to 1 and  that of the rest $N\times(1-I\left(W\right))$ coordinate channels will asymptotically tend to 0. It means $N\times I(W)$ coordinate channels are noiseless, which can be used to perfectly transmit information bits, while the rest $N\times(1-I\left(W\right))$ coordinate channels are completely noisy and can be used to transmit frozen bits. After channel polarization, we compute the channel quality of each coordinate channel $W_N^{(i)}$, and rank the coordinate channels in descending order according to their channel quality. Then we choose the former $\left\lfloor N\times I(W)\right\rfloor (\left\lfloor\cdot\right\rfloor$ means the value is rounded down) coordinate channels to transmit information bits, and the rest to transmit frozen bits.
\begin{figure}[htbp]
	\centering
	\includegraphics[width=0.4\textwidth]{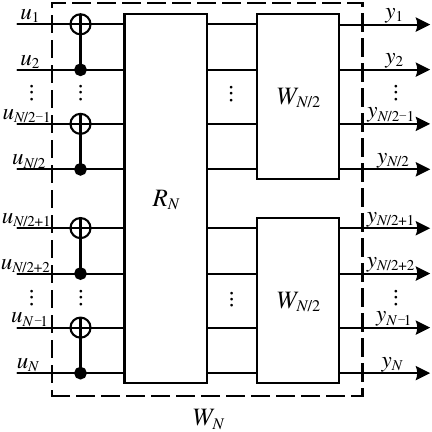}
	\caption{Classical channel polarization. $R_N$ is the reverse shuffle operation. Two independent copies of $W_{N/2}$ are combined to produce the channel $W_N$. Please see Ref. \cite{5075875} for more detail}
	\label{fig1}
\end{figure}

Arikan has proved that\cite{5075875} if the primal channel $W$ is a binary erasure channel (BEC), the channel quality of classical coordinate channels $\{W_N^{\left(i\right)}\}$ can be measured by the Bhattacharyya parameter, and they can be precisely computed through the following recursion,
\begin{equation}
	\begin{aligned}
		\label{BEC recursion}
		Z\left(W_{N}^{(2 j-1)}\right) &=2 Z\left(W_{N / 2}^{(j)}\right)-Z\left(W_{N / 2}^{(j)}\right)^{2} \\
		Z\left(W_{N}^{(2 j)}\right) &=Z\left(W_{N / 2}^{(j)}\right)^{2}
	\end{aligned}
\end{equation}
However, if the primal channel $W$ is a binary symmetric channel (BSC), there is no such a recursion to precisely compute the channel quality of classical coordinate channels $\{W_N^{\left(i\right)}\}$, only the following relationships,
\begin{equation}
	\begin{aligned}
		\label{BSC recursion}
		Z\left(W_{2 N}^{(2 i-1)}\right) & \leq 2 Z\left(W_{N}^{(i)}\right)-Z\left(W_{N}^{(i)}\right)^{2} \\
		Z\left(W_{2 N}^{(2 i)}\right) &=Z\left(W_{N}^{(i)}\right)^{2}
	\end{aligned}
\end{equation}
Since the TPM of $W_N^{(i)}$ is exponential to the code length $N$, precisely computing the Bhattacharyya parameter of these $\{W_N^{\left(i\right)}\}$ is intractable. Fortunately, ref \cite{6899669} have proved that, if the primal channel $W$ is a BSC, we can still use Eq. \eqref{BEC recursion} to approximately estimate the Bhattacharyya parameter of classical coordinate channels $\{{W}_N^{(i)}\}$.

\subsection{Polarization of two-dimensional-input quantum symmetric channels}
\label{2.4}

For  a two-dimensional-input quantum channel $\mathcal{E}$ which has a BTPM, except for its symmetric coherent information, we can also use the Bhattacharyya parameter to measure its channel quality,which is defined as
\begin{equation}
	Z(\mathcal{E}) \triangleq \sum_{\ket{i} \in B_{\text {out }}} \sqrt{\operatorname{Pr}(\ket{i}\mid \ket{0}) \operatorname{Pr}(\ket{i}\mid \ket{1})}
\end{equation}
where $B_{\text {out }}$ is the basis of output space of $\mathcal{E}$, $B_{\text{in}}=\{\ket{0},\ket{1}\}$ is the basis of input space of $\mathcal{E}$, and $\operatorname{Pr}(\ket{i}|\ket{k})$, $\ket{k} \in B_{\text {in }}$,$\ket{i} \in B_{\text {out }}$ are basis transition probabilities. 

\subsubsection{Polarization of channel quality of quantum coordinate channels}
\label{2.4.1}
As shown in Fig. \ref{fig2}, quantum channel polarization is similar to classical channel polarization and also consists of two processes -- quantum channel combining and quantum channel splitting\cite{yi2023channel}. 
We only briefly introduce some concept and conclusions here.

Through quantum channel combining, one can combine $N$ copies of primal quantum channel $\mathcal{E}:\rho^Q\rightarrow\rho^Y$ in a recursive way and obtain a quantum combined channel $\mathcal{E}_N:\rho^{Q_1\cdots Q_N}\rightarrow\rho^{Y_1\cdots Y_N}$. The difference is that we replace the XOR gates in classical channel combining by quantum CNOT gates, and use quantum SWAP gates to realize reverse shuffle operator. One can see that this process is realized by converting the classical polar coding circuits to quantum version. In quantum channel splitting process, we split the quantum combined channel $\mathcal{E}_N$ back into $N$ quantum coordinate channels $\mathcal{E}_N^{(i)}:\ \rho^{Q_i}\rightarrow\rho^{Y_1\cdots Y_N,R_1\cdots R_{i-1}}$, where $R_i$ is the reference system of $Q_i$, $1\le i\le N$. Here we use notation $\rho^{Q_1^N}$ to denote $\rho^{Q_1\cdots Q_N}$ which is similar to Arikan’s shorthand for a vector. Hence the quantum combined channel and quantum coordinate channel can be rewritten as $\mathcal{E}_N:\rho^{Q_1^N}\rightarrow\rho^{Y_1^N}$ and  $\mathcal{E}_N^{(i)}:\ \rho^{Q_i}\rightarrow\rho^{Y_1^NR_1^{i-1}}$, respectively. 

It has been proved that  if the primal quantum channel $\mathcal{E}$ is a two-dimensional-input QSC with two-dimensional output, the quantum combined channel $\mathcal{E}_N$ is a QSC, and the quantum coordinate channels $\{{\mathcal{E}}_N^{(i)}\}$ are all two-dimensional-input QQSCs\cite{yi2023channel}. Besides, if the BTPM of the primal QSC $\mathcal{E}$ and the TPM of classical primal BSC $W$ are the same, the MSLCI $I(\rho^{Q_i},\mathcal{E}_N^{(i)})$ of the quantum coordinate channel $\mathcal{E}_N^{(i)}$ equals to the Shannon capacity $I(W_N^{(i)})$ of the classical coordinate channel $W_N^{(i)}$. Since classical coordinate channels $\{W_N^{\left(i\right)}\}$ polarize, quantum coordinate channels $\{\mathcal{E}_N^{(i)}\}$ polarize as well.

\begin{figure*}[htbp]
	\centering
	\includegraphics[width=0.8\textwidth]{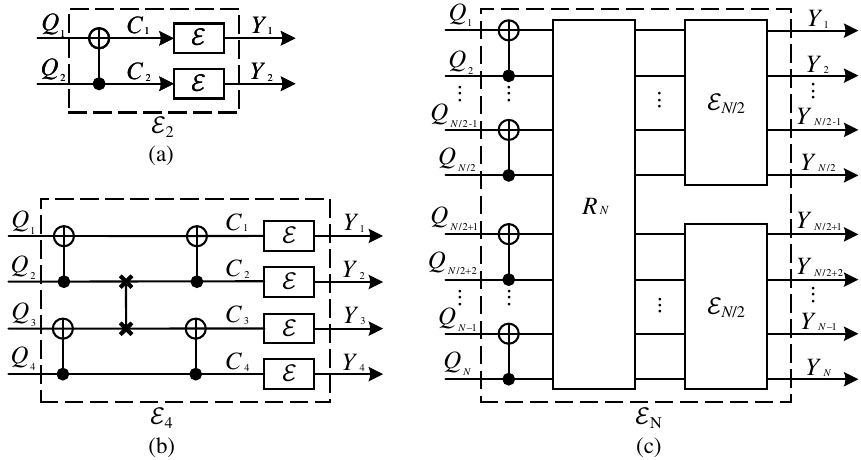}
	\caption{Quantum channel polarization circuits. (a)Two primal channel $\mathcal{E}$ combines to form channel $\mathcal{E}_2$. (b) Two $\mathcal{E}_2$ combines to form channel $\mathcal{E}_4$. (c)Two $\mathcal{E}_{N/2}$ combines to form channel $\mathcal{E}_N$, $R_N$ is the reverse shuffle operator[6].}
	\label{fig2}
\end{figure*}

In this paper, the quantum channel $\mathcal{E}$ that we consider is bit-flip channel, namely, only Pauli X error happens. The reason for this is that correcting Pauli $X$ error is the prerequisite to correct arbitrary Pauli errors. If a quantum stabilizer code doesn't work in bit-flip channel, then it will not work in a channel where arbitrary Pauli errors might happen. 

The BTPM of bit-flip channel is the same as the TPM of BSC $W$. Hence, we have $I\left(\rho^{Q_i},\mathcal{E}_N^{\left(i\right)}\right)=\ I(W_N^{(i)})$ and $Z(\mathcal{E}_N^{\left(i\right)})={Z(W}_N^{\left(i\right)})$. Same as the approximation method of Bhattacharyya parameter of classicl channels mentioned in Section 2.3, we can approximately estimate the Bhattacharyya parameter of quantum coordinate channels $\{\mathcal{E}_N^{\left(i\right)}\}$ by the following recursion,
\begin{equation}
	\begin{aligned}
		\label{QQCC recursion}
		Z\left(\mathcal{E}_{N}^{(2 i-1)}\right) &=2 Z\left(\mathcal{E}_{N / 2}^{(i)}\right)-\left(Z\left(\mathcal{E}_{N / 2}^{(i)}\right)\right)^{2}\\
		Z\left(\mathcal{E}_{N}^{(2 i)}\right) &=\left(Z\left(\mathcal{E}_{N / 2}^{(i)}\right)\right)^{2}
	\end{aligned}
\end{equation}
The greater $Z(\mathcal{E}_N^{\left(i\right)})$ is, the less reliable $\mathcal{E}_N^{(i)}$ is.

\section{The weight of stabilizer generators and logical operators of quantum polar stabilizer codes}
\label{3}
This section reveals the relationship between the channel quality of quantum coordinate channels and the weight of their corresponding stabilizer generators and logical operators.

According to the intuition that one can design a quantum polar coding scheme for quantum computing based on polarization phenomenon of quantum channels arises from directly converting the classical polar coding circuits into quantum ones, the circuits shown in Fig. \ref{fig2} might correspond to a QPSC encoding circuit.

As shown in Fig. \ref{fig2}, the process $\ket{Q_1^N}\rightarrow\ket{C_1^N}$ can be seen as an QPSC encoding circuit denoted by $E_N$, and this encoding circuit can be specified as a unitary operator $U_{enc}$. As we mention in Section \ref{2.1}, the Pauli $Z$ operators (only the Z-type operators need to be considered since only Pauli X errors are considered) acting on the frozen qubits are transformed to stabilizer generators while the Pauli operators acting on the data qubits are transformed to logical operators by $U_{enc}$. We denote the unitary operator of CNOT gate by notation $U_{CNOT}$, with qubit 1 as target and qubit 2 as control. Fig. \ref{fig3} shows that if a Pauli $Z$ operator (whose weight is 1) “pass through” the target position of a CNOT gate, its weight becomes 2, while if it “pass through” the control position, its weight remains invariant. The transformation of a Pauli X operator is right the opposite. 
The transformation of Pauli X operators is right the opposite.

\begin{figure*}[htbp]
	\centering
	\includegraphics[width=0.8\textwidth]{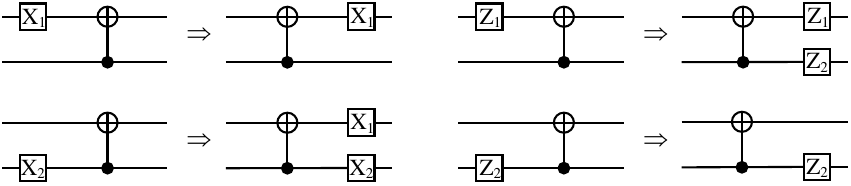}
	\caption{Transformation of Pauli $X$ operators and Pauli $Z$ operators under conjugation by CNOT gate}
	\label{fig3}
\end{figure*}

To reveal the relationship between the weight of stabilizer generators and logical operators of QPSCs and the channel quality of quantum coordinate channels, we first define two concepts "node" and "path" for the quantum polarization circuits.

\begin{definition}[\textbf{Node}]
	\label{Node}
	The operation which acts on a single qubit in a CNOT gate is defined to be a node. We use notation $\bigoplus$ to denote the target node which represents “XOR” operation of a CNOT gate, and notation $\bigodot$ to denote control node which represents identity operation of a CNOT gate.
\end{definition}

It's obvious that encoding circuit $E_N$ consists of $n=\log_2{N}$ CNOT gates blocks, and each block contains $N/2$ CNOT gates (or equivalently $N$ nodes). Arbitrary two adjacent CNOT gates blocks are connected by reverse shuffle operation realized by SWAP gates. The index of CNOT gates block is numbered from right to left, that is to say, the index of the rightmost CNOT gates block is 1, and  that of the leftmost one is $n$. The index of a node located at a CNOT gates block is numbered from top to bottom. One can see that if the index of a node is odd number, it must be a target node $\bigoplus$, if that of a node is even number, it must be a control node $\bigodot$.

\begin{definition}[\textbf{Path of a qubit passing through an encoding circuit}]
	\label{Path}
	A path $\mathcal{P}_j$ of $j$th input qubit $Q_j$ in the encoding circuit $E_N$ with length $N$ is defined as an ordered set $\left\{{node}_n^j,\cdots,{node}_1^j\right\}$, where ${node}_i^j\in\{\bigoplus,\bigodot\}$, $1\le i\le n$, $1\le j\le N$. 
\end{definition}

It's obvious that each ordered set which represents a complete path from the input of $E_N$ to the output contains $n=\log_2{N}$ nodes. Here, the reason for labeling the index $i$ of node from n to 1 rather than 1 to n is to keep this index consistent with the index of CNOT gates block (i.e., ${node}_i^j$ is located at the $i$th CNOT gates block).

Let's take the encoding circuit $E_4$ for example. As shown in Fig. \ref{fig2}(b), the path of the first qubit is $\mathcal{P}_1=\{\oplus,\oplus\}$, the second is $\mathcal{P}_2=\left\{\odot,\oplus\right\}$, the third is $\mathcal{P}_3=\left\{\oplus,\odot\right\}$, and the fourth is $\mathcal{P}_4=\left\{\odot,\odot\right\}$.

For path $\mathcal{P}_j$ in $E_N$, we observe that $\mathcal{P}_j$ corresponds to not only the $j$th stabilizer generator $S_j$ which comes from the transformation of Pauli $Z$ operator acting on qubit $Q_j$, but also the $j$th quantum coordinate channel $\mathcal{E}_N^{(j)}$ whose input is $Q_j$.

\begin{definition}[\textbf{Length of a path}]
	\label{Length of a path}
	For a path $\mathcal{P}=\left\{{node}_n,\cdots,{node}_1\right\}$, its length is defined as the number of nodes that it contains, denoted by $\left|\mathcal{P}\right|$.
\end{definition}

\begin{definition}[\textbf{Subpath}]
	\label{Subpath}
	For a path $\mathcal{P}=\left\{{node}_n,\cdots,{node}_1\right\}$ and a path $\mathcal{P}^\prime=\left\{{node}_k,{node}_{k-1}\cdots,{node}_{k-r}\right\}$, if $\mathcal{P}^\prime$ is part of $\mathcal{P}$, namely, $\mathcal{P}=\{{node}_n,\cdots,\mathcal{P}^\prime,\cdots,{node}_1\}=\left\{{node}_n,\cdots,{node}_k,{node}_{k-1}\cdots,{node}_{k-r},\cdots,{node}_1\right\}$, $\mathcal{P}^\prime$ is a subpath of path $\mathcal{P}$, denoted by $\mathcal{P}^\prime\in\mathcal{P}$. We notice that $\mathcal{P}^\prime$ follows ${node}_{k+1}$, and is followed by ${node}_{k-r-1}$. Specially, if $k=n$, $k-r>1$, we call $\mathcal{P}^\prime$ the \textbf{former path} of $\mathcal{P}$, if $k<n$, $k-r=1$, we call $\mathcal{P}^\prime$ the \textbf{latter path} of $\mathcal{P}$, and if $k=n$, $k-r=1$, we say $\mathcal{P}^\prime$ is identical to $\mathcal{P}$, denoted by $\mathcal{P}^\prime=\mathcal{P}$.
\end{definition}

For arbitrary two paths $\mathcal{P}_i$ and $\mathcal{P}_j$, as long as they have one node not the same, we say they are different, denoted by $\mathcal{P}_i\neq\mathcal{P}_j$. It's obvious that two paths with different lengths must be different.

\begin{lemma}[\textbf{Uniqueness of path}]
	\label{lemma1}
	Arbitrary two paths $\mathcal{P}_i$ and $\mathcal{P}_j$ $(1\le i\neq j\le N)$ in the same encoding circuit $E_N$ are different, namely, $\mathcal{P}_i\neq\mathcal{P}_j$. That is to say, each path in $E_N$ is unique.
\end{lemma}

\begin{proof}
	Here, we apply the mathematical induction to prove it.
	
	\textbf{Step 1}: For $n=1$, as shown in Fig. \ref{fig2}(a), the encoding circuit $E_2$ only contains one CNOT gate, it is obvious these two paths are different.
	
	\textbf{Step 2}: For $n=k\ (k>1)$, we assume that arbitrary two paths in encoding circuit $E_{N/2}\ (N/2=2^k)$ are different.
	
	\textbf{Step 3}: For $n=k+1\ (k>1)$, two independent copies of $E_{N/2}$ are combined to produce a larger encoding circuit $E_N$ in a recursive way. As shown in Fig. \ref{fig2}(c), we call the upper encoding circuit $E_{N/2}$ block 1, and the lower encoding circuit $E_{N/2}$ block 2. Here are 3 cases in total that we should consider.
	
	\textbf{Case 1}: For arbitrary two paths $\mathcal{P}_1^\prime$ and $\mathcal{P}_2^\prime$ both in block 1 or block 2, according to step 2, they are different, so we cannot find two identical paths $\mathcal{P}_1$ and $\mathcal{P}_2$ in $E_N$, which satisfies $\mathcal{P}_1^\prime\in\mathcal{P}_1$ and $\mathcal{P}_2^\prime\in\mathcal{P}_2$.
	
	\textbf{Case 2}: For two paths $\mathcal{P}_1^\prime$ in block 1 and $\mathcal{P}_2^\prime$ in block 2, if they are different, we also cannot find two identical paths $\mathcal{P}_1$ and $\mathcal{P}_2$ in $E_N$, which satisfies $\mathcal{P}_1^\prime\in\mathcal{P}_1$ and $\mathcal{P}_2^\prime\in\mathcal{P}_2$.
	
	\textbf{Case 3}: For two paths $\mathcal{P}_1^\prime$ in block 1 and $\mathcal{P}_2^\prime$ in block 2, if they are identical, according to the recursive construction of $E_N$ from two copies of $E_{N/2}$, we can see that path $\mathcal{P}_1^\prime$ follows a target node $\oplus$, and path $\mathcal{P}_2^\prime$ follows a control node $\odot$. These target node $\oplus$ and control node $\odot$ belong to the same CNOT gate. Hence, we have two different paths $\mathcal{P}_1=\{\oplus,\mathcal{P}_1^\prime\}$ and $\mathcal{P}_1=\{\odot,\mathcal{P}_2^\prime\}$.
	
	Therefore, no matter in which case, we cannot find two identical paths in the same encoding block $\mathcal{E}_N$, which completes the proof.
\end{proof}

\begin{lemma}[\textbf{Relation between two paths which have only one different node}]
	\label{lemma2}
	For two paths $\mathcal{P}_i=\left\{{node}_n^i,\cdots,{node}_1^i\right\}$ and $\mathcal{P}_j=\left\{{node}_n^j,\cdots,{node}_1^j\right\}\ (i<j)$ in encoding circuit $E_N$, the number of target node $\oplus$ contained by $\mathcal{P}_i$ is $n_\oplus^i$ and that contained by $\mathcal{P}_j$ is $n_\oplus^j$. We have $n_\oplus^i=n_\oplus^j+1$ if and only if these two paths pass through the same CNOT gate.
\end{lemma}

\begin{proof}
	(1) Necessity:
	Assuming that the $m$th ($2\leq m\leq log_{2}(n-1)$) nodes of $\mathcal{P}_i$ and $\mathcal{P}_j$ belong to the same CNOT gate, namely, $\mathcal{P}_i=\left\{\mathcal{P}_i^{for},\oplus,\mathcal{P}_i^{lat}\right\}$ and $\mathcal{P}_j=\left\{\mathcal{P}_j^{for},\odot,\mathcal{P}_j^{lat}\right\}$, where $\mathcal{P}_i^{for}$ and $\mathcal{P}_j^{for}$ are the former path of $\mathcal{P}_i$ and $\mathcal{P}_j$ respectively, and $\mathcal{P}_i^{lat}$ and $\mathcal{P}_j^{lat}$ are the latter path of $\mathcal{P}_i$ and $\mathcal{P}_j$ respectively.
	
	Since the $m$th nodes of $\mathcal{P}_i$ and $\mathcal{P}_j$ belong to the same CNOT gate, according to the recursive construction of encoding circuit, we must have $\mathcal{P}_i^{lat}=\mathcal{P}_j^{lat}$.
	
	Suppose the indexes of $\oplus$ and $\odot$ are $a$ and $b$ respectively, since they belong to the same CNOT gate, we have $b=a+1$. According to the recursive construction of encoding circuit, the node $\oplus$ with index $a$ located at the $(m+1)$th CNOT gates block of $E_N$ must following  ${node}_{m+2}$ with index $2a$ (or $2a-1$) located at the $(m+2)$th CNOT gates block, and node $\odot$ with index $b=a+1$ located at the $(m+1)$th CNOT gates block must following  ${node}_{m+2}^\prime$ with index $2b$ (or $2b-1$) located at the $(m+2)$th CNOT gates block, which means the index of ${node}_{m+2}$ and ${node}_{m+2}^\prime$ must both be odd or even numbers. If they are both even numbers, these two nodes are both control nodes, and if they are both odd numbers, these two nodes are both target nodes. Recursively using this relation, we have $\mathcal{P}_i^{for}=\mathcal{P}_j^{for}$, thus $n_\oplus^i=n_\oplus^j+1$.
	
	(2) Sufficiency: 
	Using proof by contradiction, we assume that there exists a $path_h$, which satisfies $n_\oplus^h=n_\oplus^j+1$, but $path_h$ and $path_j$ do not pass through the same CNOT gate. Suppose that the $l$th nodes of $path_h$ and $path_j$ are different, according to the necessity above, at the CNOT gates block corresponding to the $l$th node, there is a path $path_i$ that passes through the same CNOT gate as $path_j$, which satisfies $n_\oplus^i=n_\oplus^j+1$, and their former paths and latter paths are the same. Thus we have $path_h = path_i$, which contradicts Lemma
	\ref{lemma1}. Thus, the path $path_h$ above do not exist.
	
	The proof is completed.
\end{proof}

\begin{lemma}[\textbf{The weight of stabilizer generator and logical operator}]
	\label{lemma3}
	For the $i$th path $\mathcal{P}_i=\left\{{node}_n^i,\cdots,{node}_1^i\right\}$ in encoding circuit $E_N$, the number of target node $\oplus$ that it contains is denoted by $n_\oplus^i$ and that of control node $\odot$ is denoted by $n_{\odot}^i$ (we can see that $n_\oplus^i+n_{\odot}^i=\left|\mathcal{P}\right|=n=\log_2{N}$). If $i$th input qubit $Q_i$ of $E_N$ is frozen qubit and its state is set to $\ket{0}$, the weight of $i$th stabilizer generator (Z-type) $S_i$ is $2^{n_\oplus^i}$. Likewise, if it is logical qubit, the weight of ith logical $X$ operator $\overline{X_i}$ is $2^{n_{\odot}^i}$.
\end{lemma}

\begin{proof}
	SinceAs mentioned in Section \ref{2.1}, the behavior of $E_N$ can be specified as a unitary operator $U_{E_N}$. If the $i$th input qubit $Q_i$ of $E_N$ is set to state $\ket{0}$, it is a frozen qubit, and the Pauli $Z$ operator $Z_i$ acting on $Q_i$ is transformed to $i$th stabilizer generators $S_i$ under conjugation by $U_{E_N}$, namely, $S_i=U_{E_N}Z_iU_{E_N}^\dag$. If $Q_i$ is a logical qubit, the Pauli $X$ operator acting on it are transformed to logical $X$ operator under conjugation by $U_{E_N}$, namely, $\overline{X_i}=U_{E_N}X_iU_{E_N}^\dag$.
	
	If $Q_i$ is a frozen qubit, we consider how the Pauli $Z$ operator $Z_i$ transforms under conjugation by $U_{E_N}$. As shown in Fig. \ref{fig3}, if a Pauli $Z$ operator (whose weight is 1) passes through a target node $\oplus$ of a CNOT gate, its weight becomes 2, while if it passes through a control node $\odot$, its weight remains invariant. If we assume that ${node}_n^i$ of $\mathcal{P}_i$ is target node $\oplus$, $Z_i$ become $Z_iZ_{i-1}$ under conjugation by the first CNOT gate. After the conjugation, $Z_i$ continues along a subpath $\mathcal{P}_i^\prime\in\mathcal{P}_i$, while $Z_{i-1}$ along another path $\mathcal{P}_{i-1}^\prime$. However, as we prove in Lemma \ref{lemma1}, we must have $\mathcal{P}_i^\prime=\mathcal{P}_{i-1}^\prime$, since $\mathcal{P}_{i-1}^\prime$ follows a control node $\odot$ while $\mathcal{P}_i^\prime$ follows a target node $\oplus$, and these two nodes belong to a same CNOT gate. Thus, the transformation of weight of $Z_{i-1}$ along $\mathcal{P}_{i-1}^\prime$ is the same as that of $Z_i$ along $\mathcal{P}_i^\prime$, which simultaneously multiply by 2 or remain invariant dependent on the first node of $\mathcal{P}_i^\prime$ is $\oplus$ or $\odot$. Recursively using this law, the weight of $i$th stabilizer generators $S_i$ must be $2^{n_\oplus^i}$.
	
	Likewise, if $Q_i$ is a logical qubit, the analysis process is the same and the weight of $i$th logical $X$ operator $\overline{X_i}$  is $2^{n_{\odot}^i}$.
\end{proof}

\begin{theorem}[\textbf{Negative correlation between the weight of stabilizer generators and the channel quality of quantum coordinate channels}]
	\label{theorem 1}
	For two paths $\mathcal{P}_a=\left\{{node}_n^a,\cdots,{node}_1^a\right\}$ and $\mathcal{P}_b=\left\{{node}_n^b,\cdots,{node}_1^b\right\}$ in encoding circuit $E_N$, the number of target node $\oplus$ contained by $\mathcal{P}_a$ is $n_\oplus^a$ and that contained by $\mathcal{P}_b$ is $n_\oplus^b$. If $n_\oplus^a>n_\oplus^b$, the Bhattacharyya parameter of quantum coordinate channel $\mathcal{E}_N^{(a)}$ is greater than that of $\mathcal{E}_N^{(b)}$, namely, $Z\left(\mathcal{E}_N^{(a)}\right)>Z\left(\mathcal{E}_N^{(b)}\right)$.
\end{theorem}

\begin{proof}
	Since $n_\oplus^a>n_\oplus^b$, we can suppose $n_\oplus^a=n_\oplus^b+x$. Here we only consider the case that $x=1$ (i.e., $n_\oplus^a=n_\oplus^b+1$). If in this case we can prove $Z\left(\mathcal{E}_N^{(a)}\right)>Z\left(\mathcal{E}_N^{(b)}\right)$, this relation holds for arbitrary $x\geq1$.
	
	By Lemma \ref{lemma2}, when $n_\oplus^a=n_\oplus^b+1$, $\mathcal{P}_a$ and $\mathcal{P}_b$ can be rewritten as $\mathcal{P}_a=\left\{\mathcal{P}_a^{for},\oplus,\mathcal{P}_a^{lat}\right\}$ and $\mathcal{P}_b=\left\{\mathcal{P}_b^{for},\odot,\mathcal{P}_b^{lat}\right\}$, where $\mathcal{P}_a^{for}=\mathcal{P}_b^{for}$, $\mathcal{P}_a^{lat}=\mathcal{P}_b^{lat}$ and these two nodes $\oplus$ and $\odot$ belong to a same CNOT gate. Suppose this CNOT gate is located at the $\left(m+1\right)$th CNOT gates block and the index of nodes $\oplus$ located at this CNOT gates block is $i^\prime$, the index of nodes $\odot$ must be $j^\prime=i^\prime+1$. According to Eq. \eqref{QQCC recursion}, we have,
	\begin{align}
		\begin{gathered}
			Z\left(\mathcal{E}_{\frac{N}{2^{n-m-1}}}^{\left(i^{\prime}\right)}\right)=2 Z\left(\mathcal{E}_{\frac{N}{2^{n-m}}}^{\left(i^{\prime \prime}\right)}\right)-\left(Z\left(\mathcal{E}_{\frac{N}{2^{n-m}}}^{\left(i^{\prime \prime}\right)}\right)\right)^{2} \\
			Z\left(\mathcal{E}_{\frac{N}{2^{n-m-1}}}^{\left(j^{\prime}\right)}\right)=\left(Z\left(\mathcal{E}_{\frac{N}{2^{n-m}}}^{\left(i^{\prime \prime}\right)}\right)\right)^{2}
		\end{gathered}
	\end{align}
	where $i^\prime=2i^{\prime\prime}-1$ and $j^\prime=2i^{\prime\prime}$. One can verify that
	\begin{equation}
		\begin{aligned}
			&Z\left(\mathcal{E}_{\frac{N}{2^{n-m-1}}}^{\left(i^{\prime}\right)}\right)-Z\left(\mathcal{E}_{\frac{N}{2^{n-m-1}}}^{\left(j^{\prime}\right)}\right)\\
			&=2 Z\left(\mathcal{E}_{\frac{N}{2^{n-m}}}^{\left(i^{\prime \prime}\right)}\right)-\left(Z\left(\mathcal{E}_{\frac{N}{2^{n-m}}}^{\left(i^{\prime \prime}\right)}\right)\right)^{2}\\
			&-\left(Z\left(\mathcal{E}_{\frac{N}{2^{n-m-1}}}^{\left(i^{\prime \prime}\right)}\right)\right)^{2}
		\end{aligned}
	\end{equation}
	Since $Z\left(\mathcal{E}_{\frac{N}{2^{n-m}}}^{\left(i^{\prime\prime}\right)}\right)$ takes values in $\left[ 0,1 \right]$, we must have $Z\left(\mathcal{E}_{\frac{N}{2^{n-m-1}}}^{(i^\prime)}\right)-Z\left(\mathcal{E}_{\frac{N}{2^{n-m-1}}}^{\left(j^\prime\right)}\right)\geq0$, with equality holds for $Z\left(\mathcal{E}_{\frac{N}{2^{n-m}}}^{\left(i^{\prime\prime}\right)}\right)=0$ or 1.
	
	Suppose the $i$th and the $j$th nodes in the $\left(m+2\right)$th CNOT gates block connect to the $i^\prime$th node $\oplus$ and the $j^\prime$th node $\odot$ in the $\left(m+1\right)$th CNOT gates block, respectively, according to Eq. \eqref{QQCC recursion}, the index $i$ and $j$ satisfy $(i=2i^\prime,\ j=2j^\prime)$ or $(i=2i^\prime-1,\ j=2j^\prime-1)$. 
	
	\textbf{Case 1}: if $i=2i^\prime$, $j=2j^\prime$, we have
	\begin{equation}
		\begin{aligned}
			Z\left(\mathcal{E}^{(i)}_{\frac{N}{2^{n-m-2}}}\right)=Z\left(\mathcal{E}_\frac{N}{2^{n-m-1}}^{\left(i^{\prime}\right)}\right)^{2}\\
			Z\left(\mathcal{E}^{(j)}_{\frac{N}{2^{n-m-2}}}\right)=Z\left(\mathcal{E}_{\frac{N}{2^{n-m-1}}}^{\left(j^{\prime}\right)}\right)^{2}
		\end{aligned}
	\end{equation}
	
	\textbf{Case 2}: if $i=2i^\prime-1$, $j=2j^\prime-1$, we have
	\begin{equation}
		\begin{aligned}
			Z\left(\mathcal{E}^{(i)}_{\frac{N}{2^{n-m-2}}}\right)=2Z\left(\mathcal{E}_\frac{N}{2^{n-m-1}}^{\left(i^{\prime}\right)}\right)-Z\left(\mathcal{E}_\frac{N}{2^{n-m-1}}^{\left(i^{\prime}\right)}\right)^{2}\\
			Z\left(\mathcal{E}^{(j)}_{\frac{N}{2^{n-m-2}}}\right)=2Z\left(\mathcal{E}_{\frac{N}{2^{n-m-1}}}^{\left(j^{\prime}\right)}\right)-Z\left(\mathcal{E}_{\frac{N}{2^{n-m-1}}}^{\left(j^{\prime}\right)}\right)^{2}
		\end{aligned}
	\end{equation}
	Since $Z\left(\mathcal{E}_{N/2^{n-m-1}}^{(i^\prime)}\right)\geq Z\left(\mathcal{E}_{\frac{N}{2^{n-m-1}}}^{\left(j^\prime\right)}\right)$, no matter which case, we have,
	\begin{equation}
		Z\left(\mathcal{E}_{\frac{N}{2^{n-m-2}}}^{(i)}\right)>Z\left(\mathcal{E}_{\frac{N}{2^{n-m-2}}}^{(j)}\right)
	\end{equation}
	Using this recursion of Bhattacharyya parameter, one can obtain
	\begin{equation}
		Z\left(\mathcal{E}_N^{\left(a\right)}\right)>Z\left(\mathcal{E}_N^{\left(b\right)}\right)
	\end{equation}
	which completes the proof.
\end{proof}

It should be noticed that when $n_\oplus^a>n_\oplus^b$, it's uncertain that which one of $Z\left(\mathcal{E}_N^{(a)}\right)$ and $Z\left(\mathcal{E}_N^{(b)}\right)$ is greater.

Using Theorem \ref{theorem 1} and Lemma \ref{lemma3}, we can obtain Corollary \ref{corollary 1}, namely, the weight of stabilizer generators and logical X operators of QPSCs will polarize.

\begin{corollary}[\textbf{Polarization of the weight of stabilizer generators and logical $X$ operators of QPSCs}]
	\label{corollary 1}
	For the encoding circuit $E_N$ as shown in Fig. \ref{fig2}(c), the corresponding stabilizer generators of quantum coordinate channels are $S_1,\cdots,S_N$ which can be divided into $n+1\ (n=\log_2{N})$ parts according to their weight, namely,
	\begin{equation}
		S=\{S_{weight=2^0},S_{weight=2^1},\cdots,S_{weight=2^n}\}
	\end{equation}
	where $S_{weight=2^i}\in S$ is the subset of $S$, which contains all stabilizer generators whose weight is $2^i$. Likewise, the corresponding logical $X$ operators of quantum coordinate channels are $L=\{\overline{X_1},\cdots,\overline{X_N}\}$ which can be divided into $n+1$ parts according to their weight, namely,
	\begin{equation}
		L=\{L_{weight=2^0},L_{weight=2^1},\cdots,L_{weight=2^n}\}
	\end{equation}
	
	Since the channel quality will polarize, according to Lemma 3 and Theorem 1, the weight of stabilizer generators and logical $X$ operators of QPSCs will polarize.
\end{corollary}

\begin{proof}
	
	For a path $\mathcal{P}=\left\{{node}_n,\cdots,{node}_1\right\}$ in the encoding circuit $E_N$, where $n=\log_2{N}$ and ${node}_i\in{\bigoplus,\bigodot}$, $1\le i\le n$, suppose the number of target node $\bigoplus$ of $\mathcal{P}$ is $n_\oplus=x\ (0\le x\le n)$, by Lemma \ref{lemma3}, its weight is $2^x$. Since each ${node}_i\ (1\le i\le n)$ only takes values in ${\bigoplus,\bigodot}$, and all paths in $E_N$ are unique, there are $\binom{n}{x}$ different paths whose number of target node $\bigoplus$ are $x$. Moreover, these $\binom{n}{x}$ different paths correspond to $\binom{n}{x}$  stabilizer generators whose weight are all $2^x$. Hence, the set $S={S_1,\cdots,S_N}$ can be divided into $n+1$ parts according to the weight, namely, $S=\{S_{weight=2^0},S_{weight=2^1},\cdots,S_{weight=2^n}\}$.
	
	By Theorem \ref{theorem 1}, the Bhattacharyya parameter of the $\binom{n}{i}$ quantum coordinate channels whose corresponding $\binom{n}{i}$ stabilizer generators are all $2^i$ is lager than that of the $\binom{n}{j}$ quantum coordinate channels whose corresponding $\binom{n}{j}$ stabilizer generators are all $2^j$, for arbitrary $0\le j<i\le n$. As the Bhattacharyya parameter will polarize, the weight of the stabilizer will polarize.
	
	For the set of logical $X$ operators $L=\{\overline{X_1},\cdots,\overline{X_N}\}$, we have the same conclusion and the similar proof is omitted.
\end{proof}

For a stabilizer code $C$, suppose its stabilizer group is $\mathcal{S}=\langle S_1,\cdots,S_{N-k}\rangle $, there is a useful way to represent $\mathcal{S}$ by using \textbf{parity-check matrix} $H$. In bit-flip channels, all stabilizer generators of $\mathcal{S}$ only contains Pauli $Z$ operator, Hence $H$ is a $(N-k)\times N$ matrix whose rows correspond to generators $S_1$ through $S_{N-k}$. For generator $S_i$, if it contains an $Z$ on the $j$th qubit, the element $g_{i,j}$ of $H$ located at $i$th row and $j$th column is  1, while it contains an $I$ on the $j$th qubit, the element $g_{i,j}$ of $H$ located at $i$th row and $j$th column is 0. For the logical operators of $C$, we can use a similar way, namely, \textbf{logical operator matrix} to represent it.

For QPSCs with code length $N=2^n, n\geq1$, if the $i$th input qubit $Q_i$ is a frozen qubit, the $i$th element of  $1\times N$ row vector representation $r\left(Z_i\right)$ of Pauli $Z$ operator $Z_i$ acting on $Q_i$ is 1, while the rest of it are 0. The stabilizer generator $S_i$, which is transformed from $Z_i$ under conjugation by the encoding process $U_{E_N}$, is  $S_i=U_{E_N}Z_iU_{E_N}^\dag$, and its corresponding row vector representation is
\begin{equation}
	\label{22}
	r\left(S_i\right)=\ r(Z_i)G_N 
\end{equation}
where $G_N=F^{\otimes n}$, $F=\left(\begin{matrix}1&1\\0&1\\\end{matrix}\right)$. Likewise, if $Q_i$ is a logical qubit, the  Pauli $X$ operator acting on $Q_i$ is $X_i$ and corresponding row vector representation is $r(X_i)$. Then the logical $X$ operator, which is transformed from $X_i$ under conjugation by the encoding process $U_{E_N}$, is $\overline{X_i}=U_{E_N}X_iU_{E_N}^\dag$, and its corresponding row vector representation is
\begin{equation}
	\label{23}
	r\left(\bar{X_i}\right)=\ r(X_i)G_N^T 
\end{equation}

\section{Stabilizer construction algorithms of quantum polar stabilizer codes}
\label{4}
In this section, we propose two stabilizer construction algorithms of QPSCs -- CA algorithm and BS algorithm. The first algorithm is based on the ranking of channel quality of quantum coordinate channels which is similar to the code construction of classical polar code mentioned in Section \ref{2.3}. The second algorithm is based on the ranking of quantum coordinate channels according to the weight of their corresponding stabilizer generators. It's obvious that CA is the same as the construction in Ref.\cite{9366784}. We believe BS is more reasonable than the existing constructions, because they are designed in a way from the point of view of stabilizer codes, by considering the weight of logical operators and stabilizers and the code distance, which influence the correcting capability of stabilizer codes.

\subsection{Coherent-information-achieving construction }
\label{}

As mentioned in Section \ref{2.4.1}, the Bhattacharyya parameter can also be used to measure quantum channel quality, and we can approximately estimate the Bhattacharyya parameter of each $\mathcal{E}_N^{(i)}$ by Eq. \eqref{QQCC recursion}. In CA algorithm, we use the Bhattacharyya parameter to rank channel quality. The detailed process of CA is as follows,

\textbf{Step 1}: For bit-flip channel $\mathcal{E}$ with error probability $p$, compute its symmetric coherent information (i.e., MSLCI) $I\left(\rho=\frac{1}{2}\ket{0}\bra{0}+\frac{1}{2}\ket{1}\bra{1}, \mathcal{E}\right)=1-H(p)$, where $H(\cdot)$ is the Shannon entropy.

\textbf{Step 2}: Determine code length $N$ and compute the number of logical qubits $k=\left\lfloor N(1-H(p))\right\rfloor$.

\textbf{Step 3}: Use Eq. \eqref{QQCC recursion} to approximately estimate the Bhattacharyya parameter $Z(\mathcal{E}_N^{(i)})$ of each $\mathcal{E}_N^{(i)}$, then rank the coordinate channels in descending order according to their Bhattacharyya parameter and choose the former $k$ channels to transmit data qubits while the last $N-k$ channels to transmit frozen qubits which is set to state $\ket{0}^{\otimes N-k}$.

\textbf{Step 4}: Use Eq. \eqref{22} and Eq. \eqref{23} to compute the row vector representation of stabilizer generators and logical operators.

The pseudocode is shown in Algorithm \ref{algorithm1}.
\begin{algorithm}[]
	\caption{Coherent-information-achieving construction}
	\label{algorithm1}
	\LinesNumbered
	\KwIn{primal error rate $p$, code length parameter $n$}
	\KwOut{parity-check matrix $H$, logical operator matrix $H_L$}
	CodeLength $N=2^n$\;
	$k=\left\lfloor N\times (1-H(p))\right\rfloor$\;
	Compute the Bhattacharyya parameter of all quantum coordinate channels: $Z_{BP}$\;
	$ChannelOrder=Sort(Z_{BP},\ 'descend')$\;
	$FrozenQubitsIndex=ChannelOrder(1:N-k)$\;
	$LogicalQubitsIndex=ChannelOrder(N-k+1:N)$\;
	Initialize  $H=zeros(N-k,N)$, $H_L=zeros(k,N)$\;
	\For{$i=1$ to $N-k$}{
		$H\left [i,FrozenQubitsIndex[i]\right]=1$\;
	}
	\For{$i=1$ to $k$}{
		$H_L\left [i,LogicalQubitsIndex[i]\right]=1$\;
	}
	$H = H G_N$\;
	$H_L = H_L G_N^T$\;
\end{algorithm}

Under this stabilizer construction algorithm, our simulation results in Section \ref{5} show that the LERs increase with code lengths. It's obvious that the CA construction is equivalent to the construction in Ref. \cite{9366784}. According to Theorem \ref{theorem 1}, one can understand why we choose the information position by channel quality — this will guarantee us get a code with distance as large as possible. 

\subsection{Block selection construction} 
\label{}
According to Theorem \ref{theorem 1}, if the channel quality of $i$th quantum coordinate channel is worse than that of the $j$th, the weight of the corresponding stabilizer generator of the $i$th must be no less than that of the $j$th. Using this, we can rank quantum coordinate channels according to the weight of their corresponding stabilizer generators in a block-by-block way. The second algorithm called block selection is based on this ranking.

For arbitrary stabilizer code, its code distance, which is defined as the minimum weight of its logical operators, determines the number of error qubits it can correct. Increasing its code distance will help to improve its error correcting capability. Given a certain code length, more logical qubits (namely, less stabilizer generators) may lead to smaller code distance. Hence, the goal of block selection algorithm is to increase the number of stabilizer generators to increase the code distance of QPSCs. 

For bit-flip channel $\mathcal{E}$ with error probabilit $p$, under the CA algorithm, the number of logical qubits of QPSC is $k=\left\lfloor N(1-H(p))\right\rfloor$. We rank the quantum coordinate channels in descending order according to the weight of their corresponding stabilizer generators. Since there might be more than one coordinate channels with the same stabilizer generator weight, this ranking method will rank the coordinate channels block by block, and those coordiante channels in the same block have the same stabilzer generator weight. Then we assume that the weight of stabilizer generators of $\left( N-k\right)$th quantum coordinate channel is $2^x\ (0\le x\le n=\log_2{N})$. According to Lemma \ref{lemma3}, the minimum weight of logical operators is $2^{n-x}$. Notice that there are $\binom{n}{x}$ quantum coordinate channels that the weight of their corresponding stabilizer generators is equal to $2^x$. If these $\binom{n}{x}$ channels are all chosen to transmit frozen qubits, the total number of stabilizer generators will be $\sum_{m=x}^{n}\binom{n}{m}$ and the minimum weight of logical operators will increase to $2^{n-x+1}$. We can see that this can help to improve the error correcting capability of QPSCs, since it increases the number of stabilizer generators and the code distance. Notice that we chose quantum coordinate channels in a block-wise way. If a coordiante channel is chosen to transmit frozen/data qubits, the whole block it belongs to will be chosen to transmit frozen/data. Hence this algorithm is called block selection. The detailed process is as follows.

\textbf{Step 1}: For bit-flip channel $\mathcal{E}$ with error probability $p$, compute its symmetric coherent information $I(\rho=\frac{1}{2}\ket{0}\bra{0}+\frac{1}{2}\ket{1}\bra{1},\mathcal{E})=1-H(p)$, where $H(\cdot)$ is the Shannon entropy.

\textbf{Step 2}: Determine code length $N$ and compute $k^\prime=\left\lfloor N(1-H(p))\right\rfloor$. We emphasize that $k^\prime$ is not the number of logical qubits.

\textbf{Step 3}: Through Eq. \eqref{22}, compute the stabilizer generator of each quantum coordinate channel, and rank these channels in descending order according to the weight of their corresponding stabilizer generators.

\textbf{Step 4}: Find the $(N-k^\prime)$th quantum coordinate channel and computing the weight of its corresponding stabilizer generator. Suppose the weight is $2^x\ (0\le x\le n=\log_2{N})$, then choose all quantum coordinate channels that the weight of their corresponding stabilizer generators is greater than or equal to $2^x$ to transmit frozen qubits. We can see that the number of frozen qubits is $\sum_{a=x}^{n}\binom{n}{a}$.

\textbf{Step 5}: Compute the number of logical qubits $k=N-\sum_{a=x}^{n}\binom{n}{a}$, and the former $N-k$ quantum coordinate channels are used to transmit frozen qubits which are set to state $\ket{0}^{\otimes N-k}$ and the last $k$ channels to transmit logical qubits.

\textbf{Step 6}: Use Eq. \eqref{22} and Eq. \eqref{23} to compute the row vector representation of each logical operator and stabilizer generator.

The pseudocode is shown in Algorithm \ref{algorithm2}.
\begin{algorithm}[bt]
	\caption{Block selection construction}
	\label{algorithm2}
	\LinesNumbered
	\KwIn{primal error rate $p$, code length parameter $n$}
	\KwOut{parity-check matrix $H$, logical operator matrix $H_L$}
	CodeLength $N=2^n$\;
	$k^\prime=\left\lfloor N\times (1-H(p))\right\rfloor$\;
	Initialize  $S=zeros(N,N)$\;
	\For{$i=1$ to $N$}{
		$S\left [i,i\right]=1$\;
	}
	$S = SG_N$\;
	Compute the weight of all rows of $S$ $Weight$\;
	$ChannelOrder=Sort(Weight,\ 'descend')$\;
	$w=ChannelOrder[k^\prime]$\;
	$x=\log_2{w}$\;
	$k=\sum_{a=x}^{n}\binom{n}{a}$\;
	$FrozenQubitsIndex=ChannelOrder(1:N-k)$\;
	$LogicalQubitsIndex=ChannelOrder(N-k+1:N)$\;
	Initialize $H=zeros(N-k,N)$, $H_L=zeros(k,N)$\;
	\For{$i=1$ to $N-k$}{
		$H\left [i,FrozenQubitsIndex[i]\right]=1$\;
	}
	\For{$i=1$ to $k$}{
		$H_L\left [i,LogicalQubitsIndex[i]\right]=1$\;
	}
	$H = HG_N$\;
	$H_L = H_LG_N^T$\;
\end{algorithm}

One can see that given the same error probability $p$ the coding rate of CA is greater than or equal to that of BS, and the number of stabilizer generators and hence the code distance of CA is no more than those of BS. Thus, intuitively, the error correcting capability of BS algorithm is greater than or equal to that of CA algorithm. However, the simulation results in Section \ref{5} don't show this feature so obviously. This is because in the range of error probability we consider, the average number of error qubits in a single simulation test is much lower than code distance for both BS and CA. Hence, the error correcting capability of BS and CA algorithm is similar with each other.

\section{Simulation results and analysis}
\label{5}

There is an important concept of QECCs -- noise threshold. For a certain class of QECCs, when the error probability is lower than its noise threshold, we can decrease the LER of single logical qubit $LER_{lq}$ by increasing the code length. When the error probability is higher than the noise threshold, the $LER_{lq}$ will increase with the code length growth, which means this QECC fails. If a class of codes work for quantum computing, they should have such threshold. 

To test the error correcting capability of QPSCs and find out wheter they have such threshold, we perform simulations. In the range of $p$ from $1 \times 10^{-5}$ to $1 \times 10^{-2}$, under both unreliable and reliable frozen qubits assumptions, no evidence shows that the LER of single logical qubit $LER_{lq}$ of both CA and BS can be decreased by increasing the code length and we cannot find the noise threshold, which means these QPSCs don't work for quantum computing. Notice that reliable frozen qubits corresponding to reliable entanglement qubits in the Ref. \cite{9366784}.

\subsection{Simulation results}
\label{5.1}

In the simulations, we assume that the error syndrome measurement is perfect and perform $100000$ simulations for each data point. Besides, two decoders are used. The first decoders is named table-look-up decoders which is realized by creating a syndrome-to-error lookup table that exhaustively list the best recovery operation for each error syndrome. The time cost and memory space of creating such complete syndrome-to-error lookup table rapidly become intractable as code length growing. To reduce the time cost and memory space, we only create an approximate incomplete table which lists the error whose the number of error qubits is within $t_{error}$, where $t_{error}$ is a parameter to control the approximation and larger than the average number of error qubits $Np_{0}$ in a single simulation test, where $N$ is the code length, $p$ is the qubit error rate. This is because in each simulation the probability of the events that the number of error qubits doesn't equal to $Np_{0}$ tends to zero as $N$ grows. We also use another decoder called bit-flip decoder whose time cost and memory space is linear with the code length. The simulation results show that in the low error rate regime from $1 \times 10^{-5}$ to $1 \times 10^{-2}$, the error correcting performance of these two decoders is close. However, we've observed that bit-flip decoder is sensitive to the primal error rate $p$, and it may fail in the high error rate regime with code length increasing. The details of bit-flip decoder are shown in Appendix \ref{appendixA}.

Fig. \ref{fig4} shows the coding rate of CA and BS algorithms with different physical qubit error rate and code length. We can see that in the range of $p$ from $1 \times 10^{-5}$ to $1 \times 10^{-4}$, under different code length considered in the simulations, the coding rate of CA is the same as that of BS. Hence, in the range of $p$ from $1 \times 10^{-5}$ to $1 \times 10^{-4}$, we only need to perform simulations under CA algorithms.

Fig. \ref{fig5} shows the simulation results with unreliable frozen qubits and table-look-up decoder. Fig. \ref{fig6} shows the simulation results with unreliable frozen qubits and bit-flip decoder. It should be noticed that the LER of code blocks $LER_{block}$, which is the rate of code blocks with logical error, is directly obtained by counting the simulation results, while the LER of single logical qubit $LER_{lq}$ is calculated from $LER_{block}$ by the following equation. These two decoders have similar performance in the range of $p$ from $1 \times 10^{-5}$ to $1 \times 10^{-2}$, which is shown in Fig. \ref{fig7}. In Fig. \ref{fig7}, the parameter $t_{error}$ which is used to control the approximation accuracy is set to 4 in the range of $p$ from $1 \times 10^{-3}$ to $1 \times 10^{-2}$ and 2 in the range of $p$ from $1 \times 10^{-5}$ to $1 \times 10^{-3}$. The performance of table-look-up decoder under CA algorithm with different $t_{error}$ is shown in Fig. \ref{fig8}. Fig. \ref{fig9} shows the simulation results with reliable frozen qubits and table-look-up decoder. Fig. \ref{fig10} shows the simulation results with reliable frozen qubits and bit-flip decoder. As shown in Fig. \ref{fig5}, Fig. \ref{fig6}, Fig. \ref{fig9} and Fig. \ref{fig10}, in the range of $p$ from $1 \times 10^{-5}$ to $1 \times 10^{-2}$, there is no evidence showing that the LER of single logical qubit $LER_{lq}$ can be decreased by increasing the code length. There are several points seems abnormal in Fig. \ref{fig5} (d), (f), (g), (i) and Fig. \ref{fig10} (g), (i). This is because of the abrupt change of the coding rate of BS algorithms.
\begin{equation}
	\label{}
	(1-LER_{lq})^{k} = 1-LER_{block}
\end{equation}

\begin{figure*}[htbp]
	\centering
	\includegraphics[width=1\textwidth]{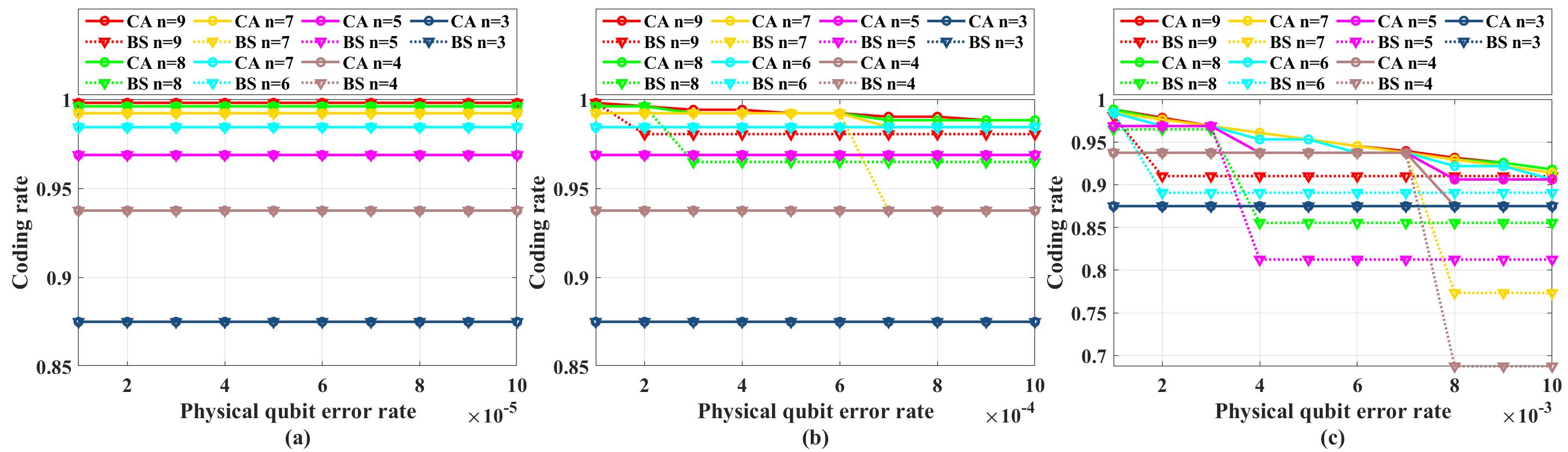}
	\caption{The coding rate of CA and BS algorithms with different physical qubit error rate and code length $N=2^n$.}
	\label{fig4}
\end{figure*}

\begin{figure*}[htbp]
	\centering
	\begin{minipage}{1\linewidth}
		\centering
		\includegraphics[width=1\linewidth]{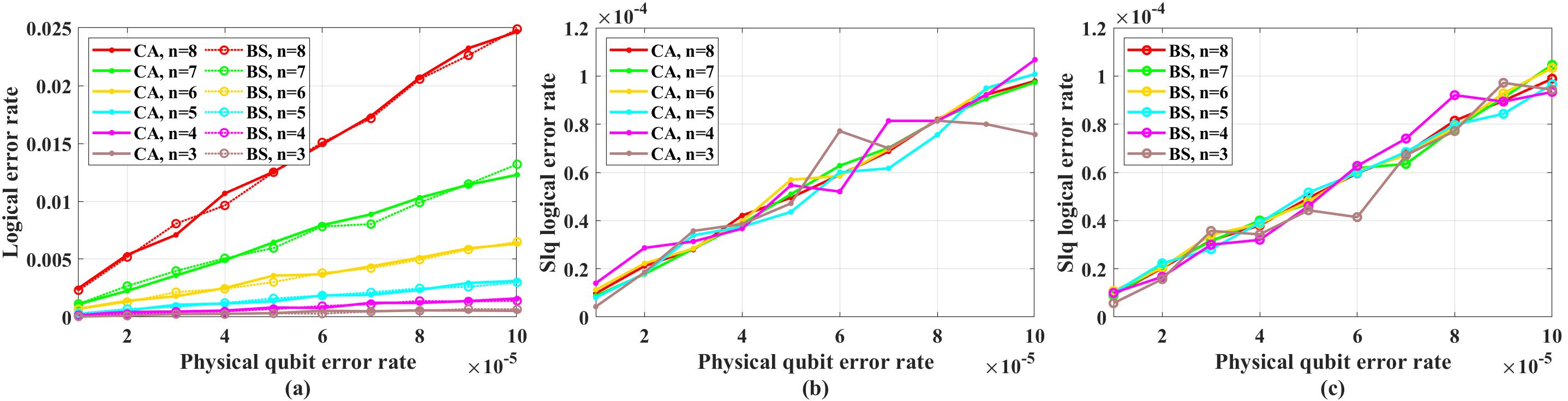}
	\end{minipage}
	\begin{minipage}{1\linewidth}
		\centering
		\includegraphics[width=1\linewidth]{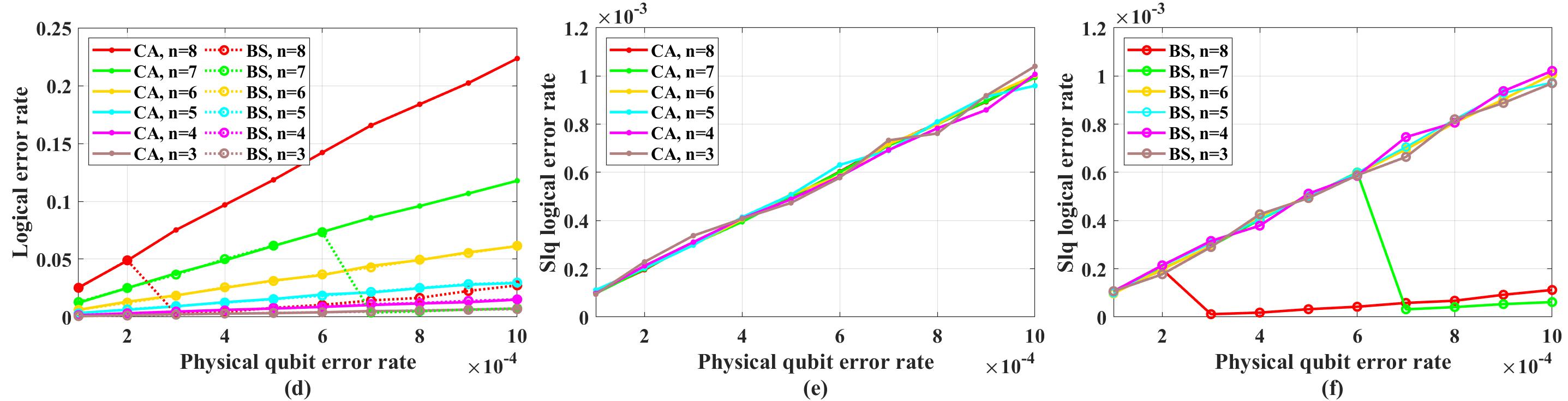}
	\end{minipage}
	\begin{minipage}{1\linewidth}
		\centering
		\includegraphics[width=1\linewidth]{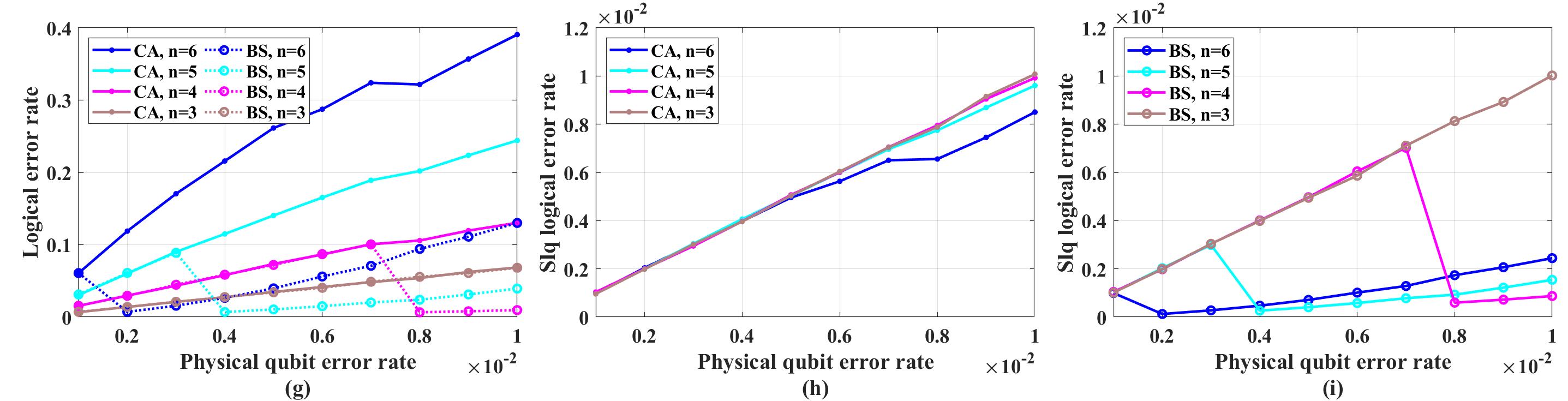}
	\end{minipage}
	\caption{The LER with unreliable frozen qubits and table-look-up decoder.}
	\label{fig5}
\end{figure*}

\begin{figure*}[htbp]
	\centering
	\begin{minipage}{1\linewidth}
		\centering
		\includegraphics[width=1\linewidth]{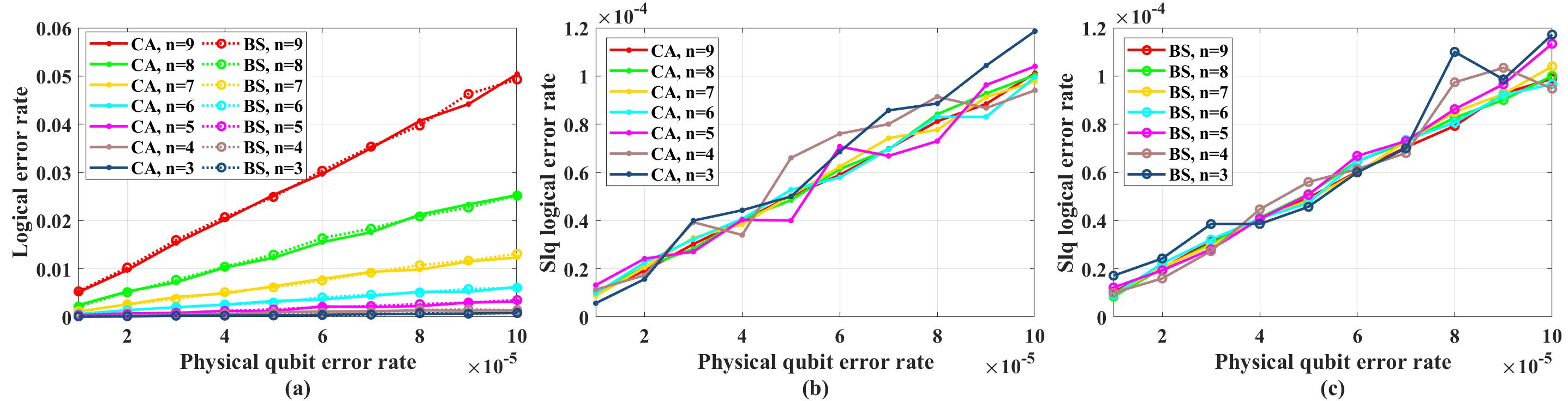}
	\end{minipage}
	\begin{minipage}{1\linewidth}
		\centering
		\includegraphics[width=1\linewidth]{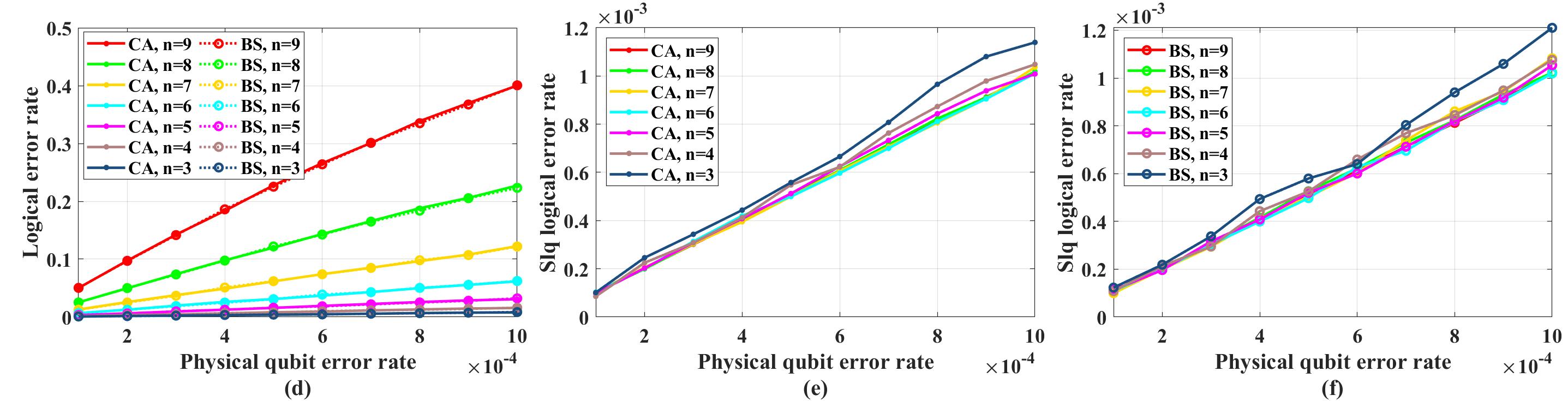}
	\end{minipage}
	\begin{minipage}{1\linewidth}
		\centering
		\includegraphics[width=1\linewidth]{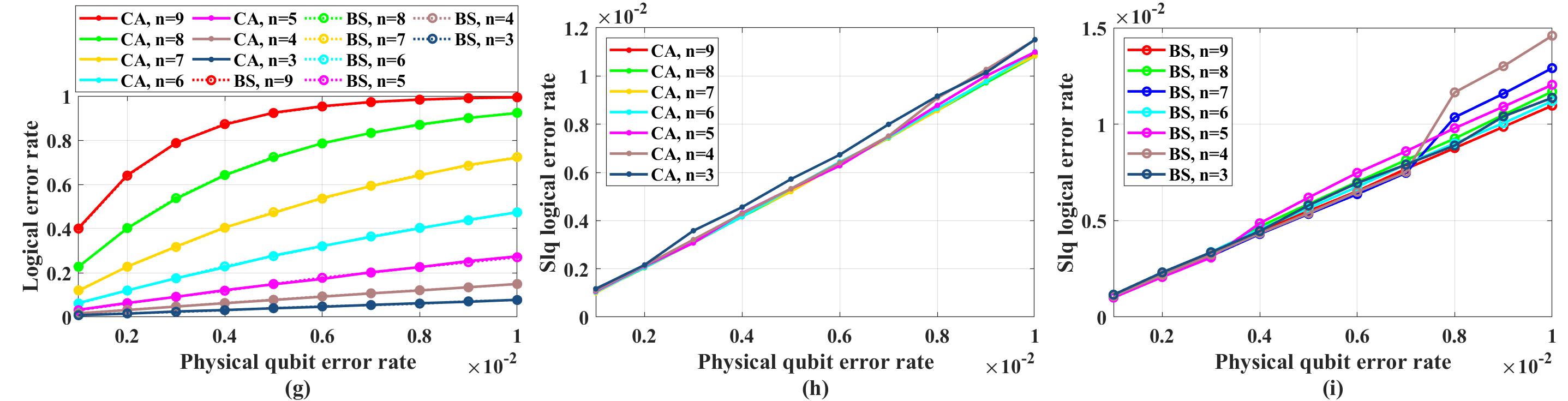}
	\end{minipage}
	\caption{The LER with with unreliable frozen qubits and bit-flip decoder.}
	\label{fig6}
\end{figure*}

\begin{figure*}[htbp]
	\centering
	\includegraphics[width=1\textwidth]{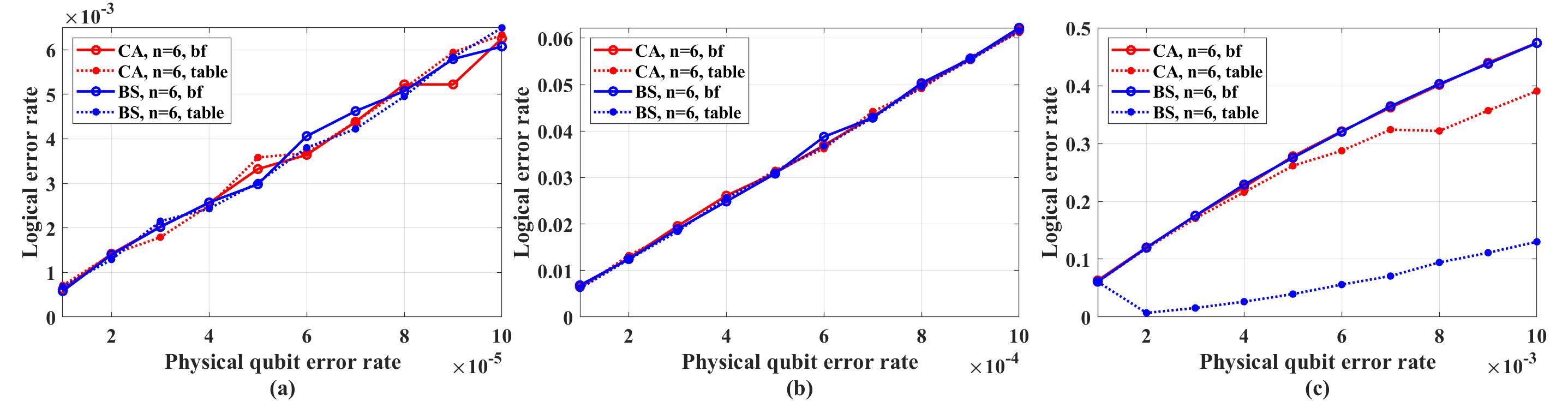}
	\caption{The performance of table-look-up decoder and bit-flip decoder in the range of $p$ from $1 \times 10^{-5}$ to $1 \times 10^{-2}$. The code length is set to $N=2^{6}$.}
	\label{fig7}
\end{figure*}

\begin{figure*}[htbp]
	\centering
	\includegraphics[width=0.7\textwidth]{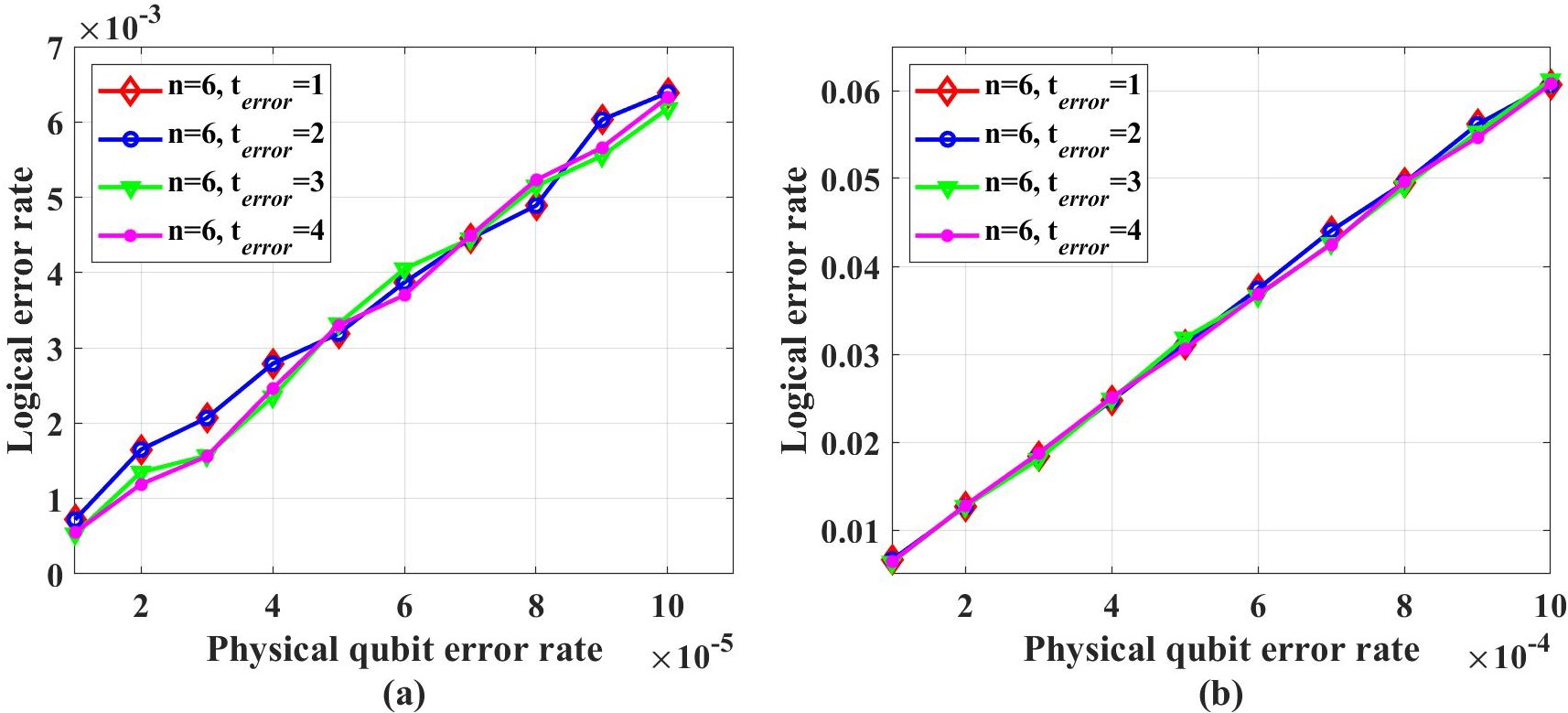}
	\caption{The performance of table-look-up decoder under CA algorithm with different $t_{error}$.}
	\label{fig8}
\end{figure*}

\begin{figure*}[htbp]
	\centering
	\begin{minipage}{1\linewidth}
		\centering
		\includegraphics[width=1\linewidth]{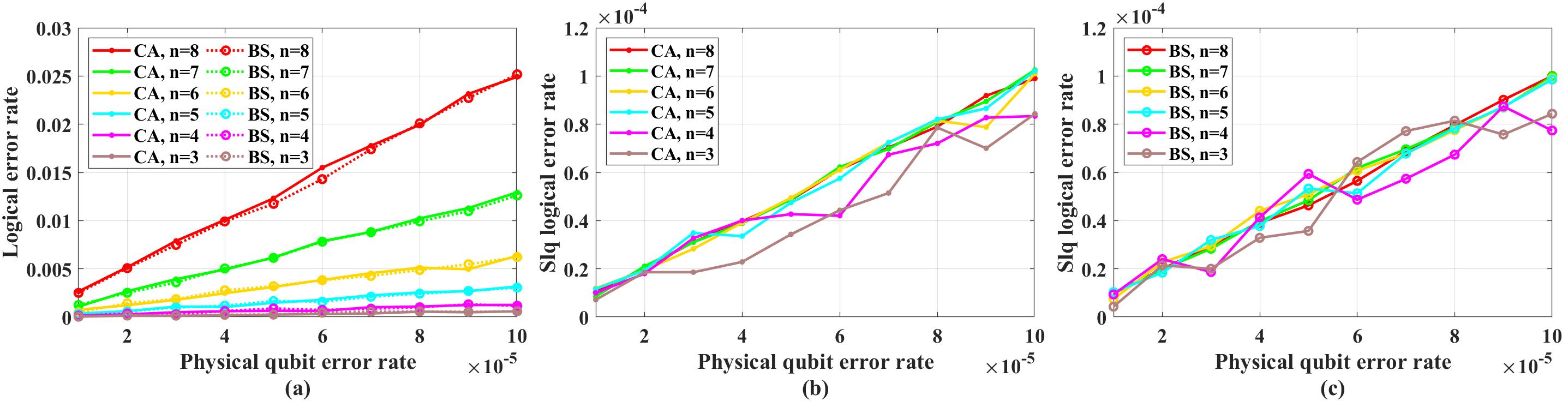}
	\end{minipage}
	\begin{minipage}{1\linewidth}
		\centering
		\includegraphics[width=1\linewidth]{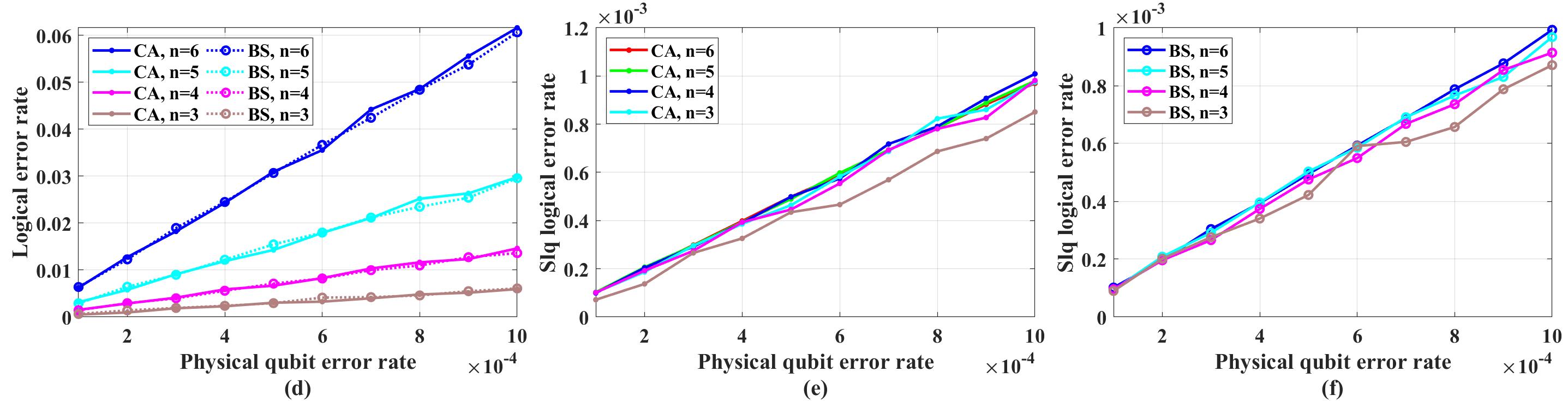}
	\end{minipage}
	\begin{minipage}{1\linewidth}
		\centering
		\includegraphics[width=1\linewidth]{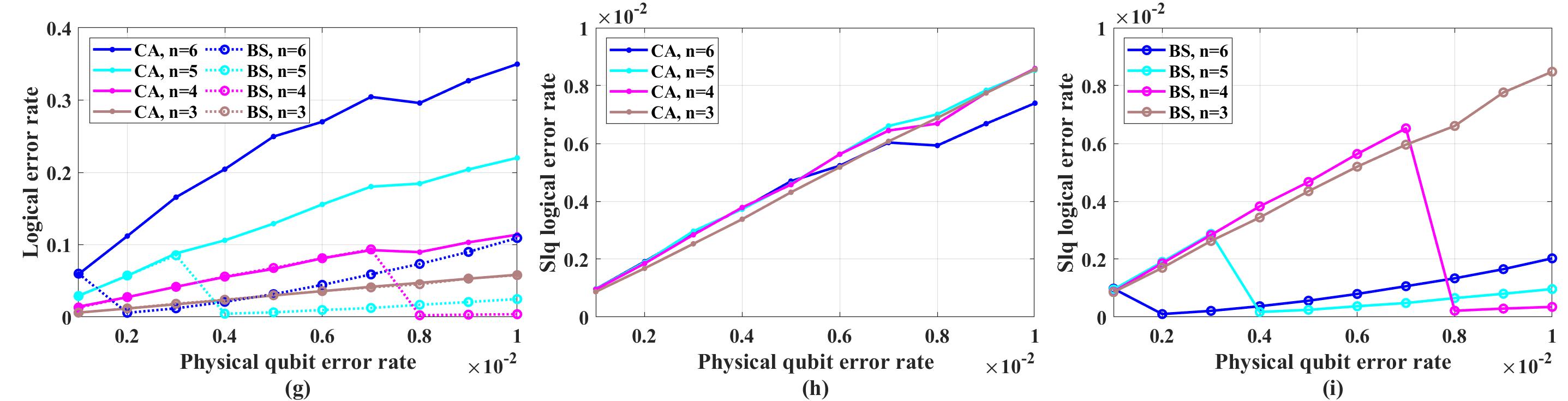}
	\end{minipage}
	\caption{The LER with reliable frozen qubits and table-look-up decoder.}
	\label{fig9}
\end{figure*}

\begin{figure*}[htbp]
	\centering
	\begin{minipage}{1\linewidth}
		\centering
		\includegraphics[width=1\linewidth]{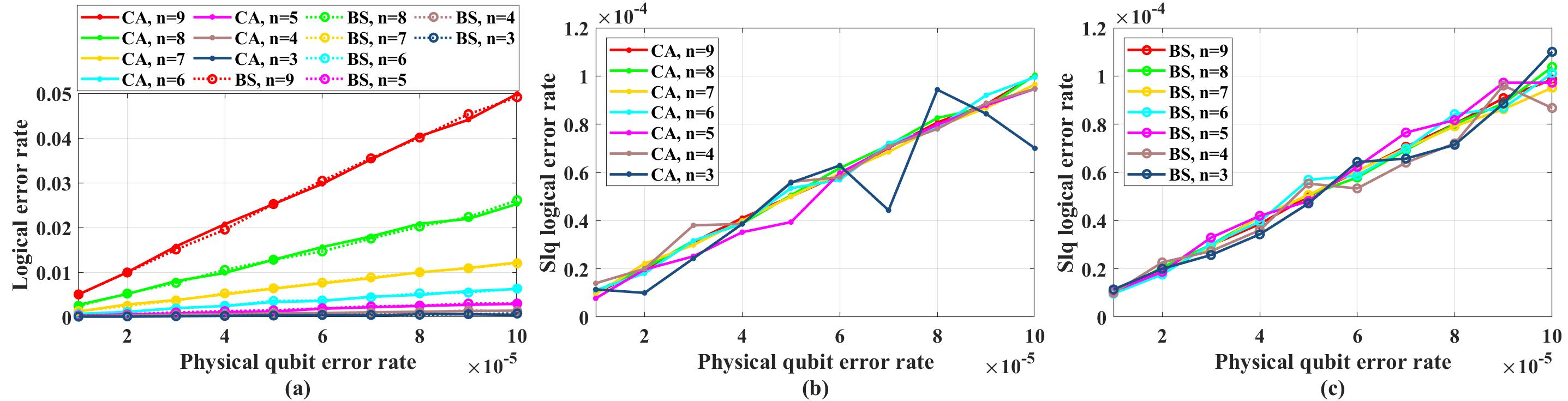}
	\end{minipage}
	\begin{minipage}{1\linewidth}
		\centering
		\includegraphics[width=1\linewidth]{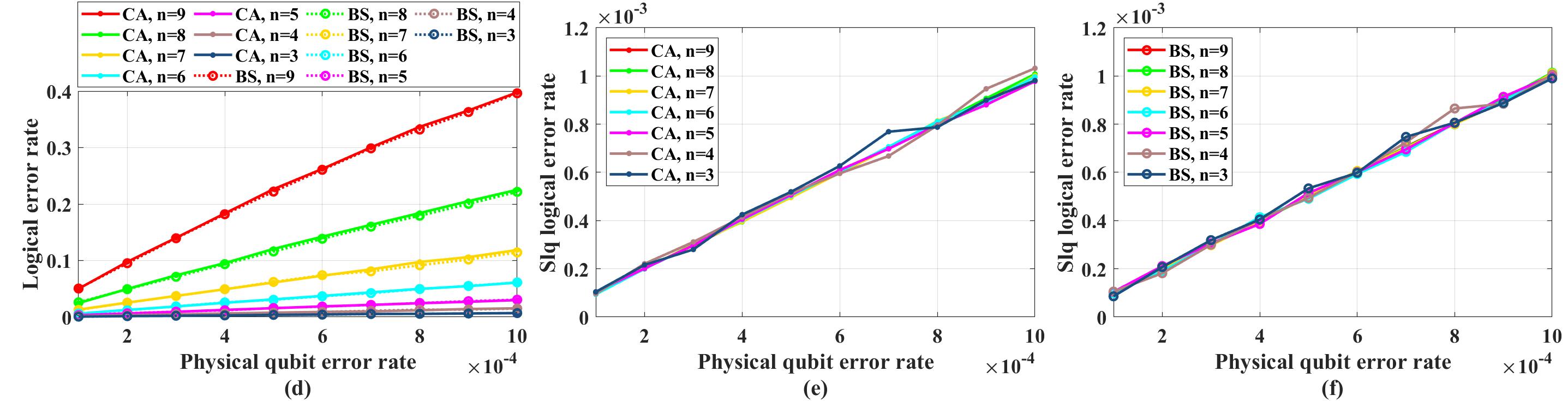}
	\end{minipage}
	\begin{minipage}{1\linewidth}
		\centering
		\includegraphics[width=1\linewidth]{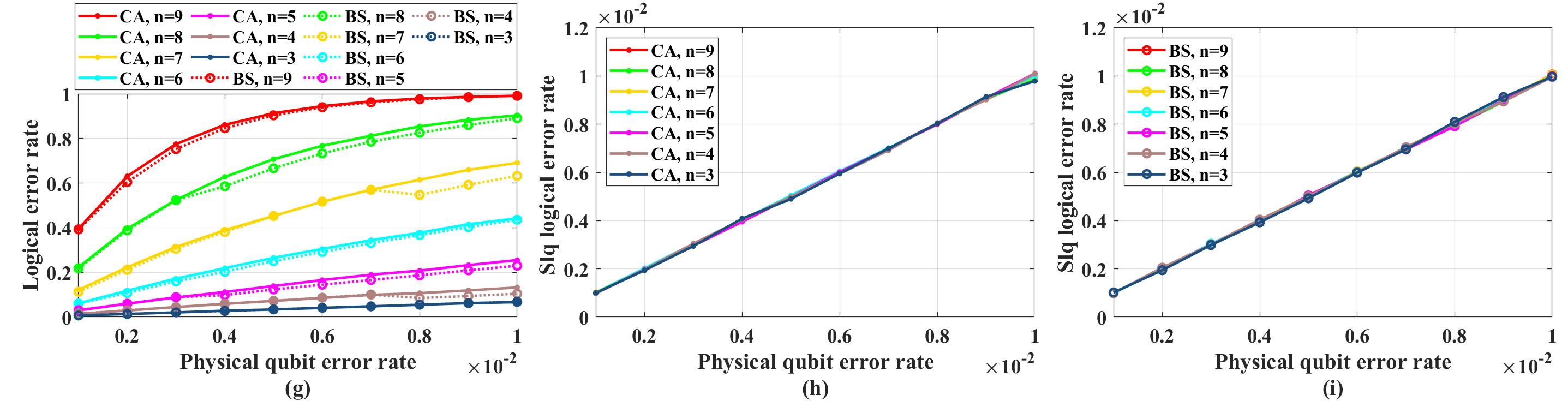}
	\end{minipage}
	\caption{The LER with reliable frozen qubits and bit-flip decoder.}
	\label{fig10}
\end{figure*}

\subsection{Analysis}
\label{5.2}
Why QPSCs don't work for quantum computing? We think this is due to the decoding channels and coordinate channels having no one-to-one relationship.

\begin{figure}[htbp]
	\centering
	\includegraphics[width=0.48\textwidth]{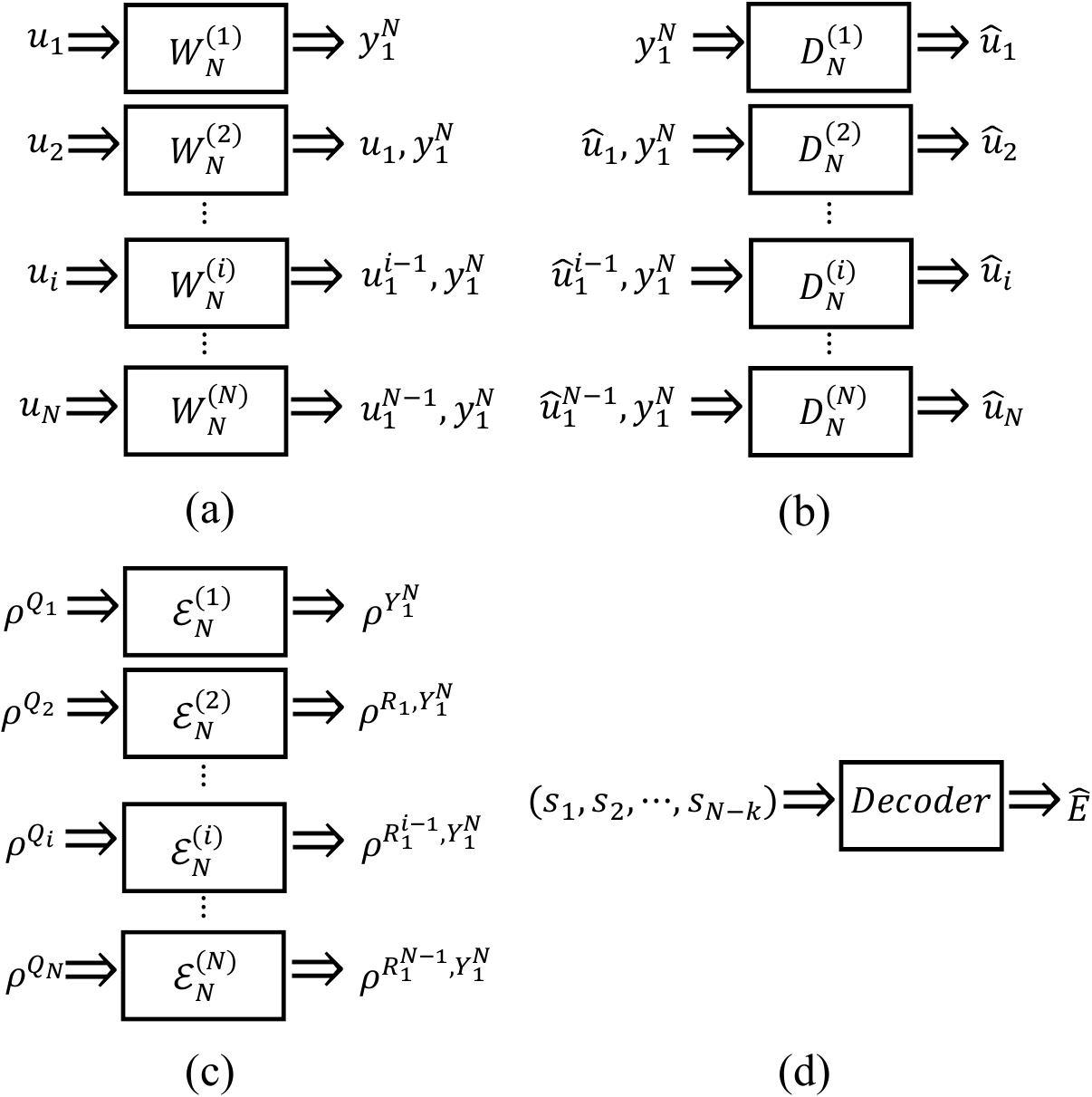}
	\caption{Coordinate channels and decoding channels. (a) Classical coordinate channels. The input of classical coordinate channel $W_N^{(i)}$ is $u_i$, and its output is $y_1^N,u_1^{i-1}$. (b) Classical decoding channels. The input of decoding channel $D_N^{(i)}$ is $\hat{u}_1^{i-1},y_1^N$ and the output is estimated $\hat{u}_i$. (c) Quantum coordinate channels. The input of quantum coordinate channel $\mathcal{E}_N^{(i)}$ is $\rho^{Q_i}$, and its output is $\rho^{Y_1^N,R_1^{i-1}}$. (d) Quantum decoding channel. The input of decoding channel is an error syndrome and the output is the most likely error.} 
	\label{fig11}
\end{figure}

For classical polar codes, as shown in Fig. \ref{fig11}(a), the input of classical coordinate channel $W_N^{(i)}$ is $u_i$, and its output is $y_1^N,u_1^{i-1}$. During the decoding procedure, as shown in Fig. \ref{fig11}(b), decoder estimates $\hat{u}_i$ after observing $y_1^N$ and the past estimated channel inputs $\hat{u}_1^{i-1}$. Hence, we can see that the classical polar code construction builds a one-to-one relationship between the classical coordinate channels and decoding channels. The better the quality of the coordinate channels are, the more reliable of the decoding channels will be.

For QPSCs, as shown in Fig. \ref{fig11}(c), the input of quantum coordinate channel $\mathcal{E}_N^{(i)}$ is $\rho^{Q_i}$, and its output is $\rho^{Y_1^N,R_1^{i-1}}$. However, for quantum decoding channel as shown in Fig. \ref{fig11}(d), its input is an error syndrome and its output is the most likely error. Both CA and BS construction algorithms don't build a similar one-to-one relationship between the quantum coordinate channels and decoding channel.

\section{Possible future direction}
\label{6}
Based on the analysis in Section \ref{5.2}, we argue that if researchers still try to borrow the idea of classical polar codes to design quantum stabilizer codes with high coding rate, the channels waiting to be analyzed and polarized should be the channels with the actual error on physical qubits as input and error syndrome as output, which corresponds to the channels in Fig. \ref{fig11}(d). What do these channels look like? The channels between the actual error on physical qubits and error syndrome can be depicted by the Tanner graph \cite{1056404} as shown in Fig. \ref{fig12}. However, it doesn’t seem direct to “polarize” the Tanner graph by borrowing the idea of classical polar codes from the perspective of channel capacity, since the number of inputs (variable nodes) and outputs (check nodes) is different and each output may connect to more than one input, which leads to the difficulty in defining the primal channel used to realize polarization and analyzing the channel capacity.
\begin{figure}[htbp]
	\centering
	\includegraphics[width=0.2\textwidth]{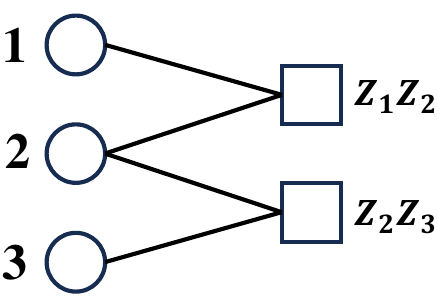}
	\caption{The Tanner graph of 3-bit-flip code.}
	\label{fig12}
\end{figure}

Arikan\cite{5075875} uses the concept of channel capacity to design the recursive encoding of classical polar codes. It does no harm to put the concept of channel capacity of classical polar codes aside for a moment and focus on the recursive encoding itself. By the recursive encoding, one can expand the encoding circuits and obtain classical polar codes with longer code length and stronger error correcting capability from two shorter classical polar codes. Hence, can we recursively expand the Tanner graph of certain stabilizer codes by the idea of recursive encoding and obtain stabilizer codes with longer code length and stronger error correcting capability. We have made some attempts in pure Pauli X noise channel and propose a class of stabilizer codes with constant coding rate 0.5 by recursively expanding the Tanner graph, which is shown in Fig. \ref{fig13}. The scheme also applies to pure Pauli Z and Y noise channel.

\begin{figure*}[htbp]
	\centering
	\includegraphics[width=0.8\textwidth]{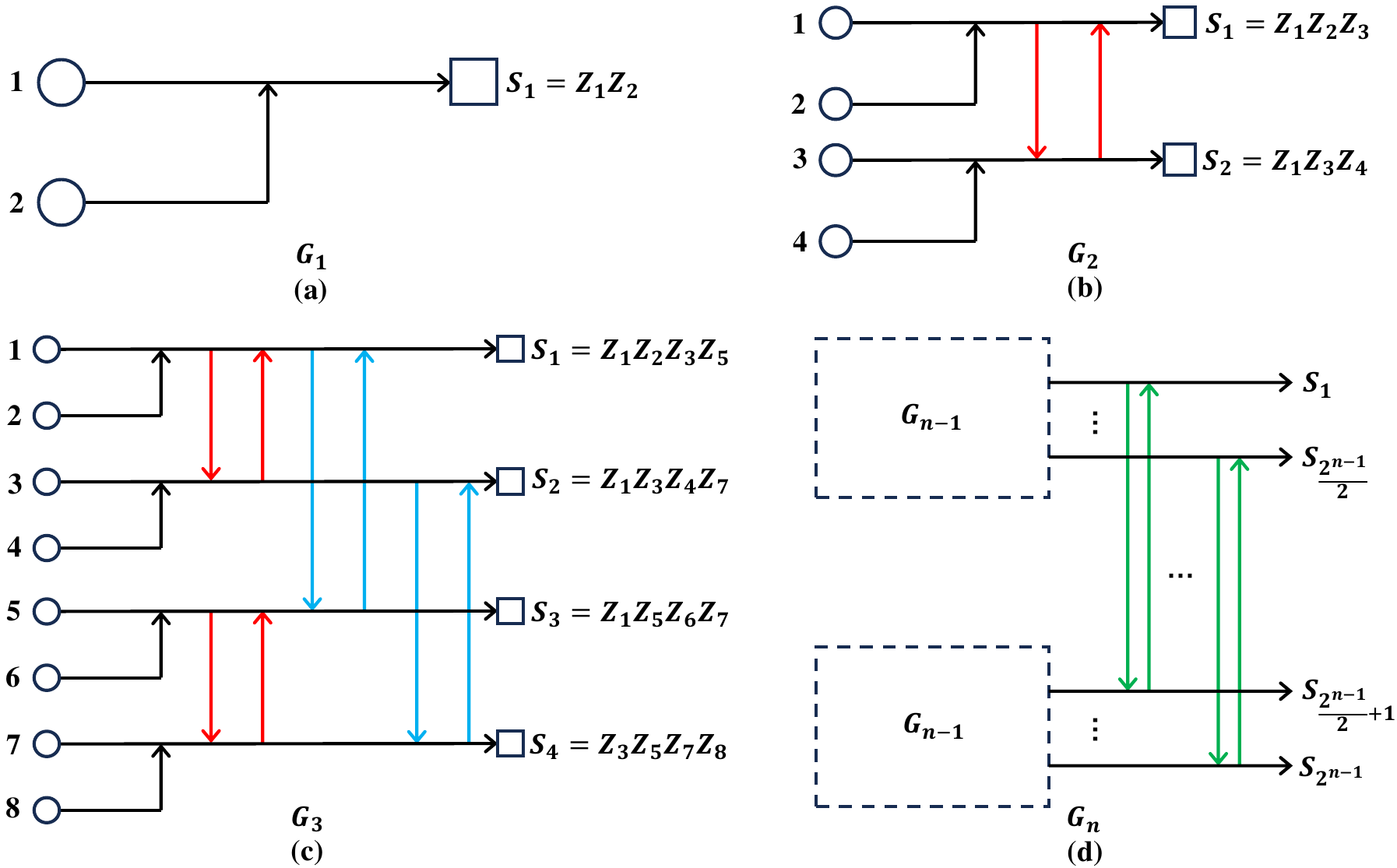}
	\caption{Quantum stabilizer codes by the recursive expansion of Tanner graph. The arrow means the corresponding qubit it starts from will join in the corresponding stabilizer it ends with. The corresponding stabilizers and logical operators are shown in Table \ref{table1}. (a) The Tanner graph $G_1$ used to recursive expansion. (b) The expanded Tanner graph $G_2$ by recursive expansion of two $G_1$. (c) The expanded Tanner graph $G_3$ by recursive expansion of two $G_2$. (d) The expanded Tanner graph $G_n$ by recursive expansion of two $G_{n-1}$}
	\label{fig13}
\end{figure*}

\begin{table*}
	\begin{center}
	\caption{The stabilizer generators and corresponding logical operators when code length $N=4, 8, 16, 32$.}
	\label{table1}
	\begin{tabular}{c|c|c|c}
		\textbf{Code length} & \ \textbf{Stabilizer generators} & \ \textbf{Logical X operators} & \ \textbf{Logical Z operators}\\
		\hline
		\multirow{2}{*}{4} & $S_1=Z_{1}Z_{2}Z_{3}$ & $\bar{X}_1=X_{1}X_{2}X_{4}$ & $\bar{Z}_1=Z_{1}$ \\
		& $S_2=Z_{1}Z_{3}Z_{4}$ & $\bar{X}_2=X_{2}X_{3}X_{4}$ & $\bar{Z}_2=Z_{3}$\\
		\hline
		\multirow{4}{*}{8} & $S_1=Z_{1}Z_{2}Z_{3}Z_{5}$ & $\bar{X}_1=X_{1}X_{2}X_{3}X_{5}$ & $\bar{Z}_1=Z_{2}$ \\
		& $S_2=Z_{1}Z_{3}Z_{4}Z_{7}$ & $\bar{X}_2=X_{1}X_{3}X_{4}X_{7}$ & $\bar{Z}_2=Z_{4}$ \\
		& $S_3=Z_{1}Z_{5}Z_{6}Z_{7}$ & $\bar{X}_3=X_{1}X_{5}X_{6}X_{7}$ & $\bar{Z}_3=Z_{6}$ \\
		& $S_4=Z_{3}Z_{5}Z_{7}Z_{8}$ & $\bar{X}_4=X_{3}X_{5}X_{7}X_{8}$ & $\bar{Z}_4=Z_{8}$ \\
		\hline
		\multirow{8}{*}{16} & $S_1=Z_{1}Z_{2}Z_{3}Z_{5}Z_{9}$ & $\bar{X}_1=X_{1}X_{2}X_{4}X_{6}X_{10}$ & $\bar{Z}_1=Z_{1}$ \\
		& $S_2=Z_{1}Z_{3}Z_{4}Z_{7}Z_{11}$ & $\bar{X}_2=X_{2}X_{3}X_{4}X_{8}X_{12}$ & $\bar{Z}_2=Z_{3}$ \\
		& $S_3=Z_{1}Z_{5}Z_{6}Z_{7}Z_{13}$ & $\bar{X}_3=X_{2}X_{5}X_{6}X_{8}X_{14}$ & $\bar{Z}_3=Z_{5}$ \\
		& $S_4=Z_{3}Z_{5}Z_{7}Z_{8}Z_{15}$ & $\bar{X}_4=X_{4}X_{6}X_{7}X_{8}X_{16}$ & $\bar{Z}_4=Z_{7}$ \\
		& $S_5=Z_{1}Z_{9}Z_{10}Z_{11}Z_{13}$ & $\bar{X}_5=X_{2}X_{9}X_{10}X_{12}X_{14}$ & $\bar{Z}_5=Z_{9}$ \\
		& $S_6=Z_{3}Z_{9}Z_{11}Z_{12}Z_{15}$ & $\bar{X}_6=X_{4}X_{10}X_{11}X_{12}X_{16}$ & $\bar{Z}_6=Z_{11}$ \\
		& $S_7=Z_{5}Z_{9}Z_{13}Z_{14}Z_{15}$ & $\bar{X}_7=X_{6}X_{10}X_{13}X_{14}X_{16}$ & $\bar{Z}_7=Z_{13}$ \\
		& $S_8=Z_{7}Z_{11}Z_{13}Z_{15}Z_{16}$ & $\bar{X}_8=X_{8}X_{12}X_{14}X_{15}X_{16}$ & $\bar{Z}_8=Z_{15}$ \\
		\hline
		\multirow{16}{*}{32} & $S_1=Z_{1}Z_{2}Z_{3}Z_{5}Z_{9}Z_{17}$ & $\bar{X}_1=X_{1}X_{2}X_{3}X_{5}X_{9}X_{17}$ & $\bar{Z}_{1}=Z_{2}$ \\
		& $S_2=Z_{1}Z_{3}Z_{4}Z_{7}Z_{11}Z_{19}$ & $\bar{X}_2=X_{1}X_{3}X_{4}X_{7}X_{11}X_{19}$& $\bar{Z}_{2}=Z_{4}$\\
		& $S_3=Z_{1}Z_{5}Z_{6}Z_{7}Z_{13}Z_{21}$ & $\bar{X}_3=X_{1}X_{5}X_{6}X_{7}X_{13}X_{21}$& $\bar{Z}_{3}=Z_{6}$\\
		& $S_4=Z_{3}Z_{5}Z_{7}Z_{8}Z_{15}Z_{23}$ & $\bar{X}_4=X_{3}X_{5}X_{7}X_{8}X_{15}X_{23}$& $\bar{Z}_{4}=Z_{8}$\\
		& $S_5=Z_{1}Z_{9}Z_{10}Z_{11}Z_{13}Z_{25}$ & $\bar{X}_5=X_{1}X_{9}X_{10}X_{11}X_{13}X_{25}$& $\bar{Z}_{5}=Z_{10}$\\
		& $S_6=Z_{3}Z_{9}Z_{11}Z_{12}Z_{15}Z_{27}$ & $\bar{X}_6=X_{3}X_{9}X_{11}X_{12}X_{15}X_{27}$& $\bar{Z}_{6}=Z_{12}$\\
		& $S_7=Z_{5}Z_{9}Z_{13}Z_{14}Z_{15}Z_{29}$ & $\bar{X}_7=X_{5}X_{9}X_{13}X_{14}X_{15}X_{29}$
		& $\bar{Z}_{7}=Z_{14}$\\
		& $S_8=Z_{7}Z_{11}Z_{13}Z_{15}Z_{16}Z_{31}$ & $\bar{X}_8=X_{7}X_{11}X_{13}X_{15}X_{16}X_{31}$& $\bar{Z}_{8}=Z_{16}$\\
		& $S_9=Z_{1}Z_{17}Z_{18}Z_{19}Z_{21}Z_{25}$ & $\bar{X}_9=X_{1}X_{17}X_{18}X_{19}X_{21}X_{25}$& $\bar{Z}_{9}=Z_{18}$\\
		& $S_{10}=Z_{3}Z_{17}Z_{19}Z_{20}Z_{23}Z_{27}$ & $\bar{X}_{10}=X_{3}X_{17}X_{19}X_{20}X_{23}X_{27}$& $\bar{Z}_{10}=Z_{20}$\\
		& $S_{11}=Z_{5}Z_{17}Z_{21}Z_{22}Z_{23}Z_{29}$ & $\bar{X}_{11}=X_{5}X_{17}X_{21}X_{22}X_{23}X_{29}$& $\bar{Z}_{11}=Z_{22}$\\
		& $S_{12}=Z_{7}Z_{19}Z_{21}Z_{23}Z_{24}Z_{31}$ & $\bar{X}_{12}=X_{7}X_{19}X_{21}X_{23}X_{24}X_{31}$& $\bar{Z}_{12}=Z_{24}$\\
		& $S_{13}=Z_{9}Z_{17}Z_{25}Z_{26}Z_{27}Z_{29}$ & $\bar{X}_{13}=X_{9}X_{17}X_{25}X_{26}X_{27}X_{29}$& $\bar{Z}_{13}=Z_{26}$\\
		& $S_{14}=Z_{11}Z_{19}Z_{25}Z_{27}Z_{28}Z_{31}$ & $\bar{X}_{14}=X_{11}X_{19}X_{25}X_{27}X_{28}X_{31}$& $\bar{Z}_{14}=Z_{28}$\\
		& $S_{15}=Z_{13}Z_{21}Z_{25}Z_{29}Z_{30}Z_{31}$ & $\bar{X}_{15}=X_{13}X_{21}X_{25}X_{29}X_{30}X_{31}$& $\bar{Z}_{15}=Z_{30}$\\
		& $S_{16}=Z_{15}Z_{23}Z_{27}Z_{29}Z_{31}Z_{32}$ & $\bar{X}_{16}=X_{15}X_{23}X_{27}X_{29}X_{31}X_{32}$& $\bar{Z}_{16}=Z_{32}$\\
		\hline
	\end{tabular}
	
	\end{center}
	
\end{table*}

As shown in Fig. \ref{fig13}, the recursive expansion of Tanner graph has the same form as the recursive encoding circuits of classical polar codes (Fig. 8 in Ref. \cite{5075875}), and the coding rate of the corresponding code is a constant 0.5. For the minimum weight of logical X operators, if $\log N$ is an even number, it is $\log N$, while if $\log N$ is an odd number, it is $\log N+1$ (notice that in pure Pauli X noise, the error correcting capability relies on the minimum weight of logical X operators, rather than the minimum weight of all logical operators, namely, the code distance). The weight of stabilizer generators is the same as that of logical X operators. 

It is easy to prove the coding rate and the minimum weight of logical X operators. The code length $N=2^n$, and notice that when $n$ is odd number, each stabilizer generator contains a unique Pauli Z operator with even subscript, while when $n$ is even number, each stabilizer generator contains a unique Pauli Z operator with odd subscript.  Hence, each stabilizer generator can't be generated by other stabilizer generators, which means that when the code length is $N=2^n$, the $2^{n-1}$ stabilizer generators in Fig. \ref{fig13}(d) are independent. Therefore, the number of logical operators are $2^n-2^{n-1}=2^{n-1}$, which means the coding rate is 0.5. As for the minimum weight of logical X operators and stabilizer generators, one can easily prove this by mathematical induction — from the code length $N=4$ to $N=2^n$

To test the error correcting capability of the codes, simulations with table-look-up decoder are performed, whose results are shown in Fig. \ref{fig14}. Fig. \ref{fig14} shows that the logical error rate will decrease with the increase of code length within the physical qubit error rate from 0.001 to 0.01. In the simulations, $t_{error}$ is set to 5. However, due to the complexity of table-look-up decoder, we haven’t found the noise threshold of these codes. To find the noise threshold, simulations need to be performed in larger physical qubit error rate. With the increase of physical qubit error rate, the average number of error qubits will increase. According to the results in Fig. \ref{fig15}, for this class of codes, when $t_{error}$ is less than ten times of the average number of error qubits, the decoding accuracy is sensitive to $t_{error}$. Hence, to perform simulations under larger physical qubit error rate, $t_{error}$ needs take larger value, which leads the required memory become intractable.

\begin{figure*}[htbp]
	\centering
	\includegraphics[width=0.7\textwidth]{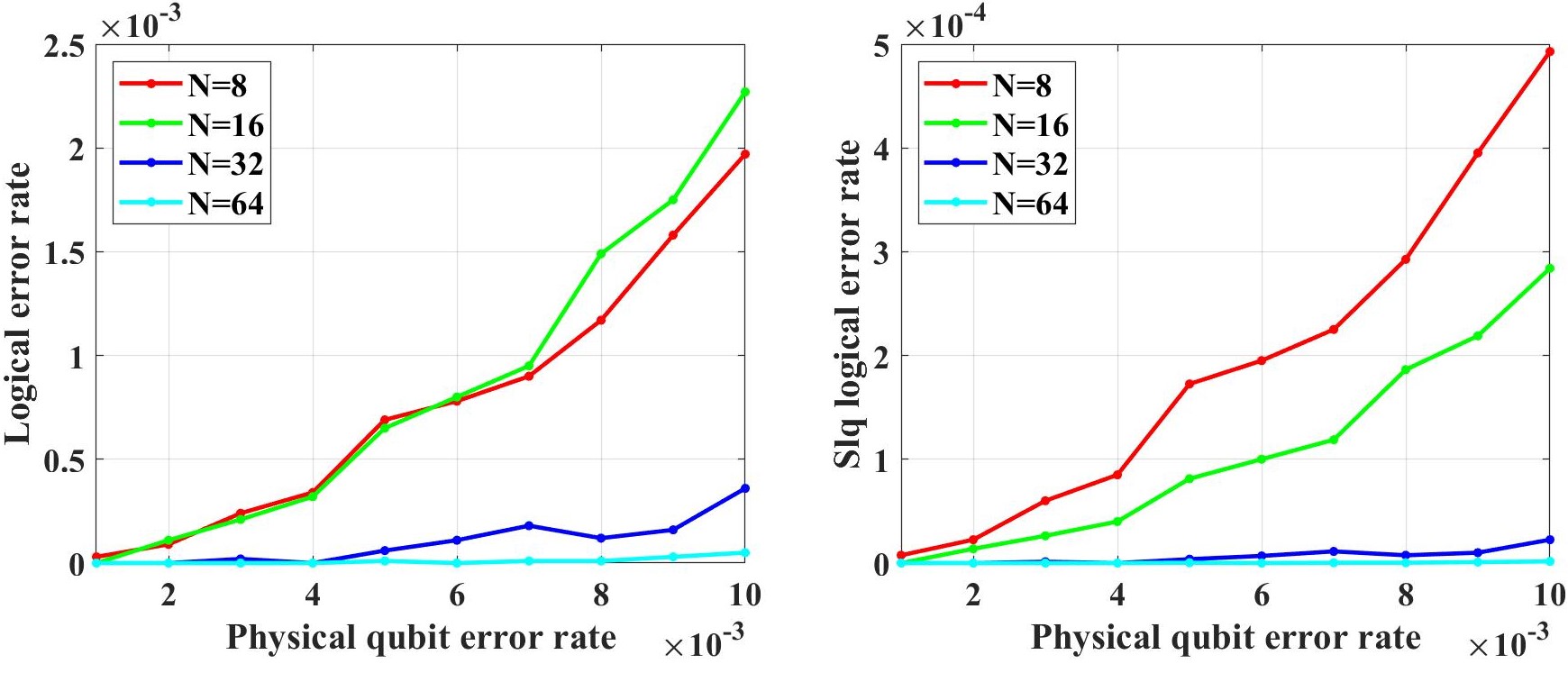}
	\caption{The LER with table-look-up decoder.}
	\label{fig14}
\end{figure*}

\begin{figure*}[htbp]
	\centering
	\includegraphics[width=0.7\textwidth]{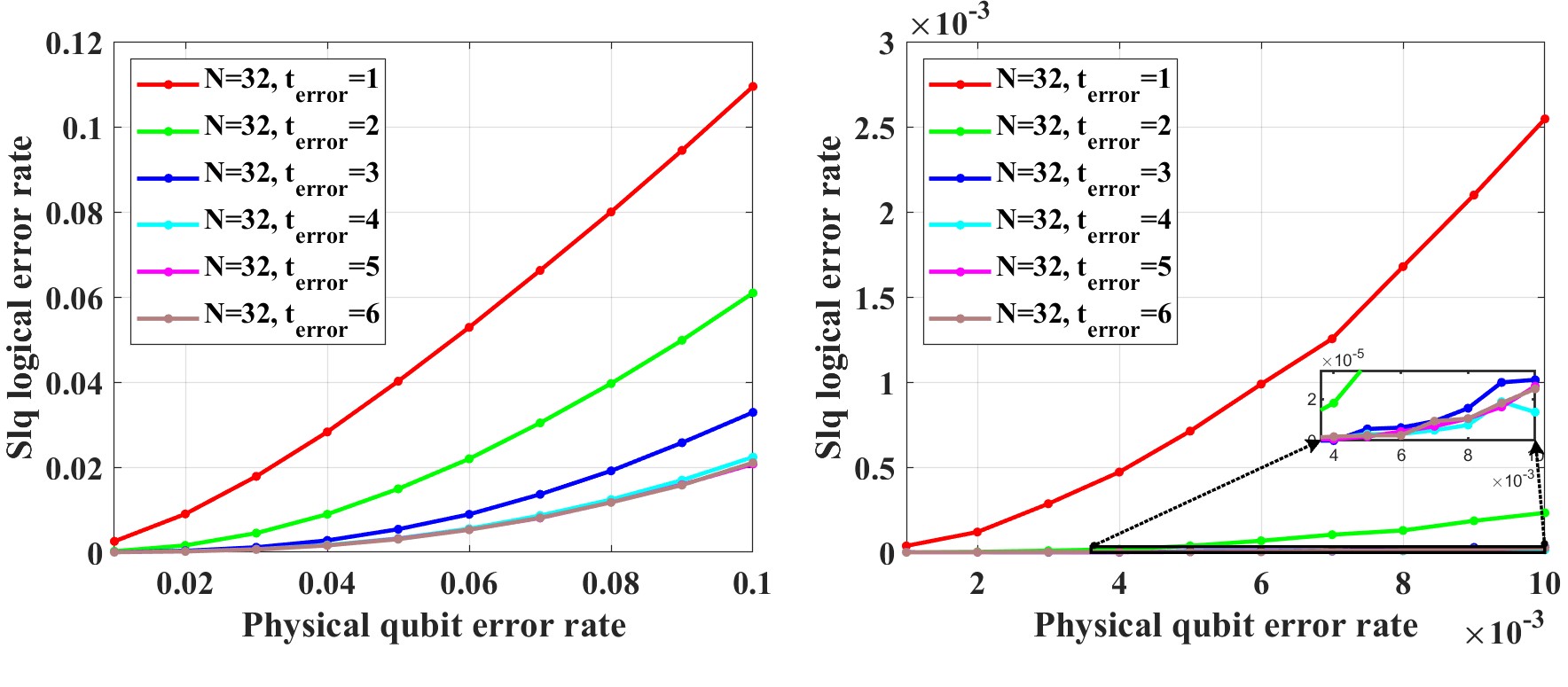}
	\caption{The decoding accuracy under different $t_{error}$.}
	\label{fig15}
\end{figure*}
The above attempts are made for pure Pauli X noise channel. It’s obvious that it also applies to pure Pauli Y noise. And simply replacing the Z operators by X operators can provide us  codes applies to pure Pauli Z noise. Whether it can be generalized to depolarizing channels is still unknown. To make this generalization, maybe one should modify the primal Tanner graph and the recursive expansion approach. However, even without this generalization, we can still obtain codes which applies to depolarizing channels by concatenating two codes – one corrects the Pauli X errors and the other corrects Pauli Z errors. The coding rate of the codes obtained by this concatenating is $1/(2\sqrt{N})$. Though the coding rate also tends to $0$ with $N$ tending to infinite, it is $\sqrt{N}$ times as much as surface codes under the same code length $N$. From the perspective of engineering realization, the greatest challenge might be the weight of the stabilizers of these codes, since the weight will increase with code length increasing. For the current hardware, the larger the weight of a stabilizer is, the lower the measurement accuracy of it will be, which will do harm to the error correction.

\section{Conclusion}
\label{7}
In this paper, we follow the intuition that one can design a quantum polar coding scheme for quantum computing based on polarization phenomenon of quantum channels arises from directly converting the classical polar coding circuits into quantum ones, and propose two possible and more reasonable stabilizer construction algorithms -- CA and BS algorithms -- to construct QSPCs. The CA construction is equivalent to the construction in Ref. \cite{9366784}. In the range of $p$ from $1 \times 10^{-5}$ to $1 \times 10^{-2}$, under both unreliable and reliable frozen qubits assumptions, no evidence shows that the LER of single logical qubit $LER_{lq}$ of both CA and BS can be decreased by increasing the code length and we cannot find the noise threshold, which means these QPSCs don't work for quantum computing. Therefore, we come to a conclusion that the above intuition is too naive to guide us to design quantum polar coding scheme for quantum computing. We also provide a possible future direction of designing quantum stabilizer codes with high coding rate by borrowing the idea of classical polar codes. Following this direction, we find a class of quantum stabilizer codes with constant coding rate 0.5 for pure Pauli X, Z and Y noise.

\section*{Acknowledgements}

This work was supported by the Colleges and Universities Stable Support Project of Shenzhen, China (No.GXWD20220817164856008), the Colleges and Universities Stable Support Project of Shenzhen, China (No.GXWD20220811170225001) and Harbin Institute of Technology, Shenzhen - SpinQ quantum information Joint Research Center Project (No.HITSZ20230111).

\bibliographystyle{quantum}
\bibliography{mycitation}

\newpage
\appendix

\section{Bit-flip decoder}
\label{appendixA}
Here we give the pseudocode of bit-flip decoder, as shown in Algorithm \ref{algorithm3}. The corresponding flow chart is shown in Fig. \ref{fig16}.

\begin{figure}[htbp]
	\centering
	\includegraphics[width=0.5\textwidth]{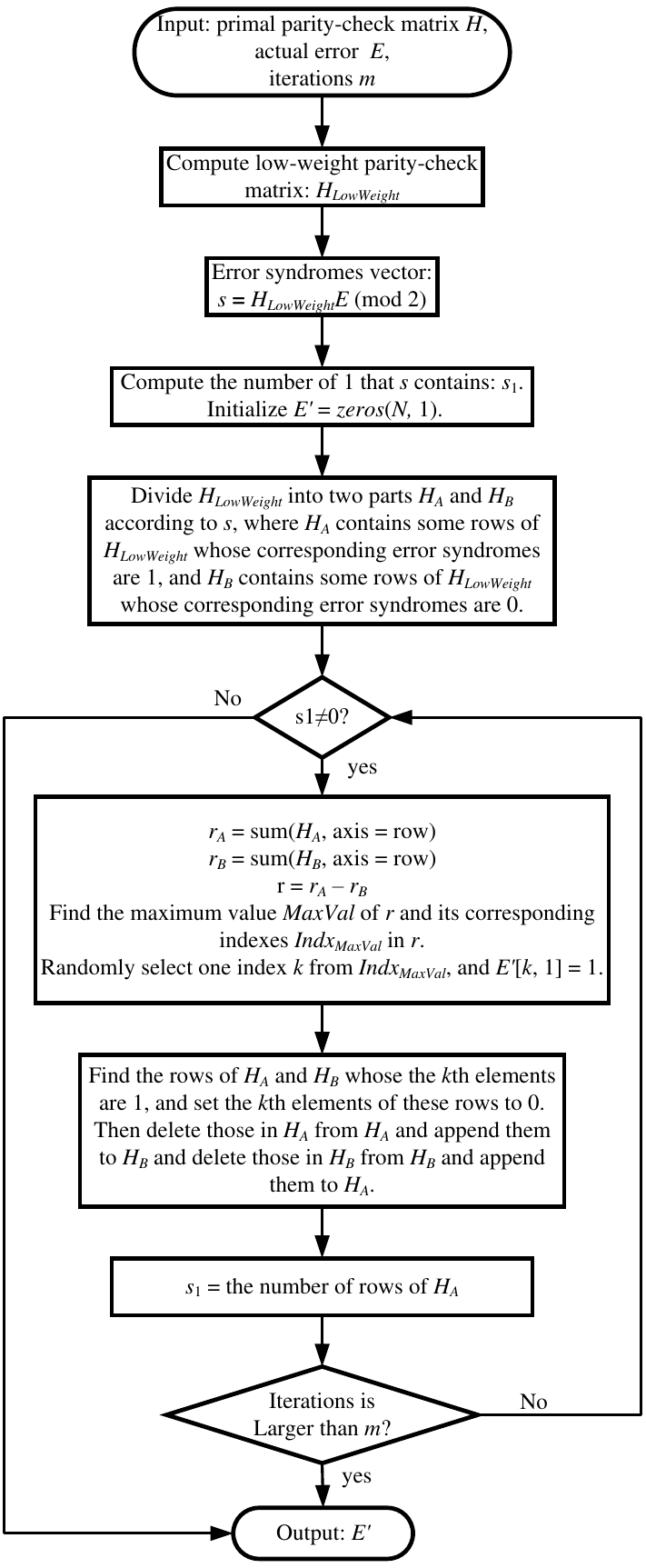}
	\caption{The flow chart corresponds to Algorithm \ref{algorithm3}}
	\label{fig16}
\end{figure}

\begin{algorithm}[htbp]
	\caption{Bit-flip decoder}
	\label{algorithm3}
	\LinesNumbered
	\KwIn{primal parity-check matrix $H$, actual error $E$, iterations $m$}
	\KwOut{estimated error $E^\prime$}
	Reduce the weight of each row of $H$ using the elementary row transformation of matrix, obtain new parity-check matrix $H_{LowWeight}$\;
	Error syndromes vector: $s=H_{LowWeight}E\ (mod\ 2)$\;
	Compute the number of 1 that $s$ contains: $s_1$. Initialize $E^\prime = zeros(N,1)$\;
	Divide $H_{LowWeight}$ into two parts $H_{A}$ and $H_{B}$ according to $s$, where $H_{A}$ contains some rows of $H_{LowWeight}$ whose corresponding error syndromes are 1, and $H_{B}$ contains some rows of $H_{LowWeight}$ whose corresponding error syndromes are 0\;
	\While{$s_1\neq 0$ or $m \neq 0$}{
		$r_A = sum(H_A,\ axis = row)$\;
		$r_B = sum(H_B,\ axis = row)$\;
		$r = r_A-r_B $\;
		Find the maximum value $MaxVal$ of $r$ and its corresponding indexes $Indx_{MaxVal}$ in $r$\;
		Randomly select one index $k$ from $Indx_{MaxVal}$\;
		$E^\prime[k,1]=1$\;
		\For{each row $H_A^i\in H_A$}{
			\If{$H_A^i[k]=1$}{
				$H_A^i[k]=0$\;
				Delete $H_A^i$ from $H_A$ and append $H_A^i$ to $H_B$\;
			}
		}
		\For{each row $H_B^j\in H_B$}{
			\If{$H_B^j[k]=1$}{
				$H_B^j[k]=0$\;
				Delete $H_B^j$ from $H_B$ and append $H_B^j$ to $H_A$\;
			}
		}
		$s_1 = the\ number\ of\ rows\ of\ H_A$\;
		$m = m-1$\;
		Output $E^\prime$\;
	}
\end{algorithm}

\end{document}